\documentclass{article}
\pdfoutput=1 

\usepackage{arxiv}

\usepackage[utf8]{inputenc} 
\usepackage[T1]{fontenc}    
\usepackage{hyperref}       
\usepackage{url}            
\usepackage{booktabs}       
\usepackage{amsfonts}       
\usepackage{bbold}
\usepackage{nicefrac}       
\usepackage{microtype}      
\usepackage{lipsum}
\usepackage{graphicx}
\usepackage{tikz}
\usepackage{algpseudocode}
\graphicspath{ {./images/} }

\usepackage{bm,amsmath,bbm,amsfonts,amsthm,amsbsy,amscd,amsxtra,amsgen,amsopn,amssymb,mathtools,caption}
\usepackage{threeparttable}
\usepackage{natbib}

\newcommand{\bs}[1]{\boldsymbol{#1}}
\newtheorem{theorem}{Theorem}

\newtheorem{condition}{Condition}
\newtheorem{definition}{Definition}

\title{Multivariate group sequential tests for global summary statistics}

\author{Abigail J. Burdon$^{1,*}$ and Thomas Jaki$^{1,2}$ \\
$^{1}$MRC Biostatistics Unit, University of Cambridge, Robinson Way, Cambridge, CB2 0SR, U.K\\
$^{2}$University of Regensburg, Bajuwarenstrasse 4, 93053 Regensburg, Germany \\
$^{*}$\texttt{Email: abigail.burdon@mrc-bsu.cam.ac.uk}}

\begin{document}
\maketitle
\begin{abstract}
We describe group sequential tests which efficiently incorporate information from multiple endpoints allowing for early stopping at pre-planned interim analyses. We formulate a testing procedure where several outcomes are examined, and interim decisions are based on a global summary statistic. An error spending approach to this problem is defined which allows for unpredictable group sizes and nuisance parameters such as the correlation between endpoints. We present and compare three methods for implementation of the testing procedure including numerical integration, the Delta approximation and Monte Carlo simulation. In our evaluation, numerical integration techniques performed best for implementation with error rate calculations accurate to five decimal places. Our proposed testing method is flexible and accommodates summary statistics derived from general, non-linear functions of endpoints informed by the statistical model. Type 1 error rates are controlled, and sample size calculations can easily be performed to satisfy power requirements.

\end{abstract}

\section{Introduction}
Clinical trials which address multiple research questions have gained traction in recent years due to their efficient nature and ability to improve drug development. \citep{dmitrienko2013key} show that multiplicity problems can arise from a variety of sources such as the evaluation of multiple endpoints, treatments or subgroup comparisons. As a result, the likelihood of encountering false-positive results is significantly increased and this has been highlighted as a major concern by the European Medicines Agency (EMA [\citeyear{ema2002multiple}]). \citep{d1993strategies} give an overview of statistical techniques known as multiple testing methods which play a pivotal role in maintaining the integrity and reliability of trial outcomes when multiple statistical tests are performed.

The source of multiplicity, the associated trial objectives and the available information about the joint distribution of the quantities causing the multiplicity are key indicators which can help to identify the most appropriate multiplicity adjustment. We focus on a multiplicity problem that most commonly arises in the context of multiple endpoints and we present a global testing procedure to assess the overall treatment effect across the endpoints. For a multivariate parameter $\bs{\theta}$ and a univariate summary statistic defined by $\Delta(\bs{\theta})$, global testing procedures present us with a framework for testing the null hypothesis $H_0:\Delta(\bs{\theta})=0$. It is well known that global testing procedures are more efficient than those which test endpoints individually (see~\citet{dmitrienko2009multiple}, \citet{pan2013multiple} and~\citet{tang1993design}). In the Discussion section of this paper, we show that power can be improved by 45\% when employing our global testing procedure compared to utilising a Bonferroni adjustment. Although not immediately obvious, trials which consider efficacy and toxicity outcomes, early readouts and delayed responses and surrogate outcomes alongside primary endpoints share this common source of multiplicity and could benefit from global testing methods (see \citep{conaway1996designs}, \citep{hampson2013group}, \citep{fleming1996surrogate}). Recent guidance by the Food and Drug Administration (FDA [\citeyear{guidance2017multiple}]) highlights the necessity of controlling type 1 error rates over all key endpoints and this will be a focus of this paper.

A widely used global testing procedure is Hotelling's $T^2-$test (\citet{hotelling1992generalization}) which is a matrix generalization of the two-sample $t-$test and compares the means of two independent Gaussian random samples. This method relies on the assumption that the alternative hypothesis $H_A$ is two-sided, which is not always the most desirable configuration. To overcome this restriction, \citet{o1984procedures} introduced the ordinary least squares (OLS) and generalized least squares (GLS) methods. Unfortunately, OLS and GLS methods assume equal effect sizes per dimension of the multivariate data and assume a given directional relationship between elements of $\bs{\theta}$ and the global statistic $\Delta(\bs{\theta})$. Finally, the approximate likelihood ratio test by~\citep{tang1989approximate} generalizes the O'Brien tests by restricting attention to the maximum test statistic which removes the directional relationship assumption but is still subject to the assumption of equal effect sizes. In the context of clinical trials with multiple endpoints, there is often useful information about the functional form of the global treatment effect encapsulated in the statistical model and this should drive the choice of summary statistic. The currently available methods only allow for linear combinations which may not be the most appropriate in all cases. Examples of statistical models which could result in non-linear global summary statistics include binary endpoints and random effect models. 

Group Sequential Tests (GSTs) allow potential early stopping for efficacy and futility at preplanned interim analyses. GSTs are therefore particularly desirable due to reduced realized sample sizes and shorter times to reach a decision. However, this presents another source of multiplicity which must be adjusted for and requires knowledge of the distribution of successive test statistics across analyses. A fundamental property for each of the currently available global testing methods is that a linear generalization is applied to calculate the global summary statistic. Distributional properties of the global treatment effect are then derived and the appropriate statistical inference performed. \citet{tang1989design} observe that repeated significance testing can automatically be applied for the O'Brien GLS approach since normality holds. However, for other global testing procedures, it is surprisingly rare that the resulting test statistic is normally distributed so cannot easily be used in the group sequential setting.

For practical implementation of global testing procedures, simulation based designs have gained approval by regulators in the fixed sample setting. In Section~\ref{sec:results}, we show that when multiple analyses are concerned, we observe a ``propagation of error" effect in which the type 1 and type 2 error rates and poorly estimated beyond an acceptable level for a trial with 5 analyses. To combat this, we propose to use a multivariate numerical integration method. A compelling feature of this work is the development of a method which results in exact error rate calculations but does not increase the amount of computation effort. 

We develop a testing procedure which incorporates the efficiencies of both multiple endpoints and GSTs. A novel feature of our proposed method is that summary statistics can be general non-linear functions of the parameters which also covers the linear case for completeness. Further, an advantage over competitor methods is that our approach only requires knowledge of the distribution of treatment effect estimates and does not rely on a distributional assumption about the global summary statistic itself. Instead, the distribution of successive multivariate parameter estimates across analyses is used to ensure that type 1 error rates are controlled and power is maximized. Our proposed method gives rise to flexible clinical trial designs where information from multiple endpoints can be leveraged to improve clinical drug development.

\section{Multivariate group sequential tests}

\subsection{Set-up and notation}
The aim of a clinical trial is to assess how effective a new experimental treatment performs compared to an existing standard-of-care drug or placebo. Often, we make statistical inferences based on some ``treatment effect measure" which is defined at the design stage of the trial. We shall often use alternative keywords such as ``outcome" and `` endpoint" to refer to the treatment effect measure. In many cases, this measure is straightforward to define and corresponds to a single parameter in a statistical model. However, this can be more challenging when there is a treatment effect parameter for each of two or more endpoints and these multiple treatment effect parameters  contribute to the overall effect of treatment. Let $\bs{\theta}=(\theta_1,\dots,\theta_p)^T$ be a vector of parameters in a statistical model. Throughout, we shall use the subscript notation $\theta_j$ to denote the $j^{th}$ dimension of this $(p\times 1)$ vector. Suppose that $\Delta(\bs{\theta})$ is the function which returns a scalar output summarizing the overall effect of treatment. In a group sequential trial (GST), let $K$ denote the total number of analyses and we shall make inferences at analyses $k=1,\dots,K$. At analysis $k,$ let $\hat{\bs{\theta}}^{(k)}$ be the parameter estimate for $\bs{\theta}$ in the statistical model and let $\Sigma^{(k)}$ be the variance-covariance matrix associated with $\hat{\bs{\theta}}^{(k)}$ for $k=1,\dots,K$. In some instances, $\Sigma^{(k)}$ will be considered known, however it is more realistic that the variance-covariance matrix will be estimated using the data which we denote $\hat{\Sigma}^{(k)}$. The dimensions of these objects are such that $\hat{\bs{\theta}}^{(k)}$ is a $(p\times 1)$ vector and $\Sigma^{(k)}$ and $\hat{\Sigma}^{(k)}$ are $(p\times p)$ matrices.

\subsection{Motivating example}
\label{subsec:example}
Throughout this paper, we shall consider an example clinical trial for cardiovascular disease (see \citet{kim2020randomized}). Let $\mathbf{X}_i = (X_{1i},X_{2i}), i=1,\dots,n$ be the vector of blood pressure and cholesterol responses for patient $i$ receiving the control treatment and let $\mathbf{Y}_i=(Y_{1i},Y_{2i}), i=1,\dots,n$ be the blood pressure and cholesterol responses for patient $i$ on the experimental treatment. Note here that the patients are recruited to the control and treatment arms in a 1:1 ratio. We suppose the responses are distributed according to 
$\mathbf{X}_i \sim N (\bs{\mu}_C, M/2)$ and $\mathbf{Y}_i \sim N (\bs{\mu}_T, M/2)$ where $\bs{\mu}_C$ and $\bs{\mu}_T$ are $2\times 1$ vectors and $M$ is a $2\times 2$ matrix. The parameter $\theta_1 = \mu_{C1} - \mu_{T1}$ represents the expected reduction in blood pressure attributed to treatment and $\theta_2 = \mu_{C2}-\mu_{T2}$ represents the reduction in cholesterol. \citet{sundstrom2018synergistic} review clinical trials in cardiovascular disease and conclude that the two endpoints have a multiplicative effect on cardiovascular health. Hence $\Delta(\bs{\theta})=\theta_1\theta_2$ is a suitable global summary for the treatment effect. We may however be concerned about the value of $\theta_1\theta_2$ when $\theta_1 < 0$ and $\theta_2 < 0$ so that we consider the global summary statistic given by the function
\begin{equation}
\label{eq:example}
\Delta(\bs{\theta})=
\begin{cases}
\theta_1\theta_2 &\theta_1 \geq 0 \cup \theta_2 \geq 0 \\
-\theta_1\theta_2 & \text{otherwise}
\end{cases}.
\end{equation}

We shall estimate $\bs{\theta}=(\theta_1,\theta_2)$ in the usual way for normally distributed responses. In a GST with $K=5$ analyses, at analysis $k$, let $n^{(k)}$ be the cumulative number of patients on each treatment arm in the trial. Then let
\begin{equation}
\label{eq:example_est}
\hat{\bs{\theta}}^{(k)} = \frac{1}{n^{(k)}}\sum_{i=1}^{n^{(k)}} \mathbf{X}_i - \mathbf{Y}_i \hspace{1cm} \text{ for } k=1,\dots,K.
\end{equation}
The variance-covariance matrix associated with $\hat{\bs{\theta}}^{(k)}$ is known to be $\Sigma^{(k)} = M/n^{(k)}$ and can be estimated by
$$
\hat{\Sigma}^{(k)} = \frac{1}{n^{(k)}(n^{(k)}-1)}\sum_{i=1}^{n^{(k)}} (\mathbf{X}_i-\mathbf{Y}_i-\hat{\bs{\theta}}^{(k)})(\mathbf{X}_i-\mathbf{Y}_i-\hat{\bs{\theta}}^{(k)})^T.
$$

\subsection{Multivariate group sequential test (MGST)}
\label{subsec:testing_procedure}
We aim to design and perform a hypothesis test which encapsulates information about multiple treatment effect parameters in a single framework. Suppose that $\mathcal{N} \subset \mathbb{R}^p$ denotes the space of treatment parameters which represents no difference between treatment and control groups. Figure~\ref{fig:reject_regions} shows an example of the null region which, with $p=2$, is given by $\mathcal{N}=\{\bs{\theta};\theta_1 \leq 0 ,\theta_2 \leq 0\}$. We shall test the one-sided hypothesis
$$
H_0:\bs{\theta} \in \mathcal{N} \hspace{1cm} \text{vs} \hspace{1cm} H_A: \bs{\theta} \notin \mathcal{N}.
$$
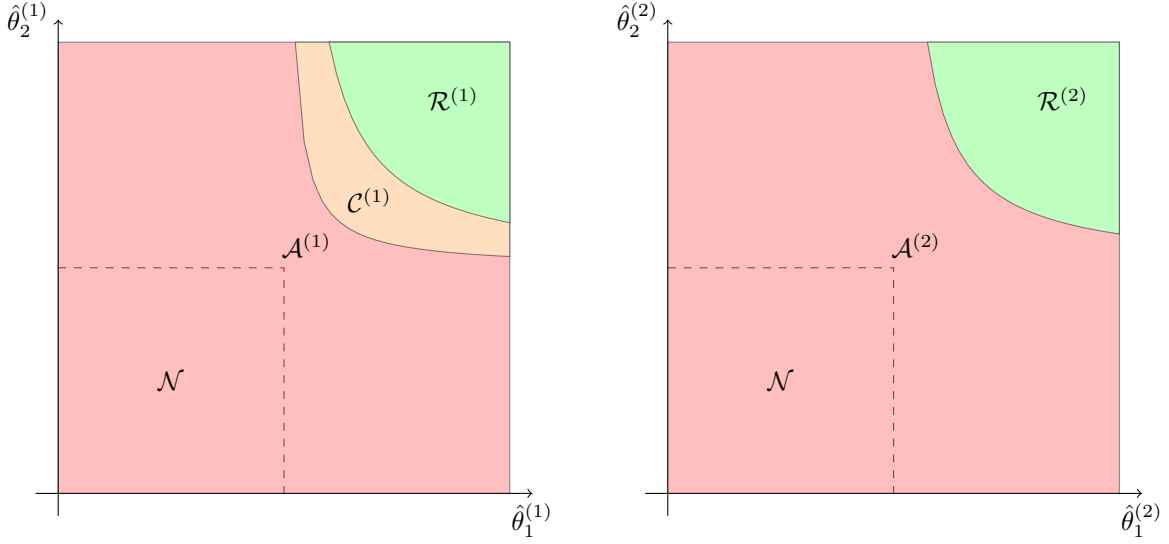
\begin{figure}[t]
\centering

\begin{tikzpicture}[scale = 0.3]

  \begin{scope}[shift={(0,0)}]
    \draw[->] (-11,-10) -- (11,-10) node[below] {$\hat{\theta}_1^{(1)}$};
    \draw[->] (-10,-11) -- (-10,11) node[left] {$\hat{\theta}_2^{(1)}$};
    \draw[dashed] (-10,0) -- (0,0) -- (0,-10);

    \pgfmathsetmacro{\ak}{5}
    \pgfmathsetmacro{\bk}{20}
    \draw[fill=red!50, opacity=0.5] (-10,-10) -- (10,-10) -- (10,10) -- (-10,10) -- cycle;
    \fill[white] plot[domain={\ak/10}:{10}] (\x, {\ak/\x}) -- (10,10) -- cycle;
    \draw[fill=orange!50, opacity=0.5] plot[domain={\ak/10}:{10}] (\x, {\ak/\x}) -- (10,10) -- cycle;
    \fill[white] plot[domain={\bk/10}:{10}] (\x, {\bk/\x}) -- (10,10) -- cycle;
    \draw[fill=green!50, opacity=0.5] plot[domain={\bk/10}:{10}] (\x, {\bk/\x}) -- (10,10) -- cycle;

    \node at (-5,-5) {$\mathcal{N}$};
    \node at (1,1) {$\mathcal{A}^{(1)}$};
    \node at (7.5,7.5) {$\mathcal{R}^{(1)}$};
    \node at (3.8,3) {$\mathcal{C}^{(1)}$};
  \end{scope}
  
  \begin{scope}[shift={(27,0)}]
    \draw[->] (-11,-10) -- (11,-10) node[below] {$\hat{\theta}_1^{(2)}$};
    \draw[->] (-10,-11) -- (-10,11) node[left] {$\hat{\theta}_2^{(2)}$};
    \draw[dashed] (-10,0) -- (0,0) -- (0,-10);
    
    \pgfmathsetmacro{\bk}{15}
    \draw[fill=red!50, opacity=0.5] (-10,-10) -- (10,-10) -- (10,10) -- (-10,10) -- cycle;
    \fill[white] plot[domain={\bk/10}:{10}] (\x, {\bk/\x}) -- (10,10) -- cycle;
    \draw[fill=green!50, opacity=0.5] plot[domain={\bk/10}:{10}] (\x, {\bk/\x}) -- (10,10) -- cycle;
    
    \node at (-5,-5) {$\mathcal{N}$};
    \node at (1,1) {$\mathcal{A}^{(2)}$};
    \node at (7.5,7.5) {$\mathcal{R}^{(2)}$};
  \end{scope}
\end{tikzpicture}

\caption{Null region $\mathcal{N}$ and decision regions $\mathcal{A}^{(k)}, \mathcal{R}^{(k)}$ and $\mathcal{C}^{(k)}$ for $k=1,2$ in a group sequential trial with $K=2$ analyses.}
\label{fig:reject_regions}
\end{figure}

In order to carry out this hypothesis test, we present a set of rules which determine when $H_0$ should be rejected or accepted. The underlying concept for this work is that the value of $\Delta(\bs{\theta})$
determines the magnitude of the global treatment effect across all endpoints and pathways. The function $\Delta(\bs{\theta})$ is a naturally arising feature of the design problem and we would like to exploit this useful information. Therefore, we shall test the hypotheses by finding parameter estimates, $\hat{\bs{\theta}}^{(k)},$ and global statistics $\Delta(\hat{\bs{\theta}}^{(k)})$ for $k=1,\dots,K.$ The proposed method is hereto referred to as the ``multivariate group sequential test" (MGST) and is now described. At analysis $k,$ let $\mathcal{A}^{(k)} \subset \mathbb{R}^{p}$ be the acception region, $\mathcal{R}^{(k)} \subset \mathbb{R}^p$ be the rejection region and $\mathcal{C}^{(k)} \subset \mathbb{R}^p$ be the continuation region which are defined by
\begin{align*}
\mathcal{A}^{(k)} &= \{\hat{\bs{\theta}}^{(k)} : \Delta(\hat{\bs{\theta}}^{(k)}) <  a^{(k)}\}\\
\mathcal{R}^{(k)} &= \{\hat{\bs{\theta}}^{(k)} : \Delta(\hat{\bs{\theta}}^{(k)}) \geq b^{(k)}\} \\
\mathcal{C}^{(k)} &= \{\hat{\bs{\theta}}^{(k)} : a^{(k)} \leq \Delta(\hat{\bs{\theta}}^{(k)}) < b^{(k)}\}
\end{align*}
for some scalars $a^{(k)}\leq b^{(k)}$ for $k=1,\dots,K.$ We note that for each $k,$ the regions $\mathcal{A}^{(k)}, \mathcal{R}^{(k)}$ and $\mathcal{C}^{(k)}$ are disjoint and their union is equal to the state space $\mathbb{R}^p.$

For the example function $\Delta(\bs{\theta})$ in Equation~\eqref{eq:example} and a total of $K=2$ analyses, Figure~\ref{fig:reject_regions} shows a visual representation of these decision regions. Whenever endpoints are tested separately while accounting for correlation, this is a specific case of our flexible testing procedure with rectangular decision regions. This has been proposed and implemented by~\citet{jennison1993group} for the trade-off between response and toxicity in the co-primary endpoint setting. \citet{conaway1996designs} go one step further than this, finding quadrilateral regions and allude to the idea that even more complex decision regions may be desirable.

In parallel to the univariate group sequential test described by~\citet{jennison2000group}, the multivariate group sequential testing procedure is as follows:

\hspace{1cm} After analysis $k=1,\dots,K-1$

\hspace{2cm} if $\hat{\bs{\theta}}^{(k)} \in \mathcal{A}^{(k)}$ \hspace{1cm} stop, accept $H_0$

\hspace{2cm} if $\hat{\bs{\theta}}^{(k)} \in \mathcal{R}^{(k)}$ \hspace{1cm} stop, reject $H_0$

\hspace{2cm} otherwise \hspace{1.5cm} continue to analysis $k+1$

\hspace{1cm} after analysis $K$

\hspace{2cm} if $\hat{\bs{\theta}}^{(K)} \in \mathcal{A}^{(K)}$ \hspace{1cm} stop, accept $H_0$

\hspace{2cm} if $\hat{\bs{\theta}}^{(K)} \in \mathcal{R}^{(K)}$ \hspace{1cm} stop, reject $H_0$

\noindent where $a^{(K)}=b^{(K)}$. This restriction implies that $\mathcal{C}^{(K)}=\emptyset$ and this is to ensure that the trial terminates at the final analysis. With the testing procedure in place, the aim is then to determine the scalar constants $a^{(k)}$ and $b^{(k)}$ for $k=1,\dots,K$ to satisfy power and type 1 error requirements.

\subsection{Multivariate canonical joint distribution (MCJD)}
For subsequent results to hold, we require the parameter estimates to satisfy the following distribution which we define as the ``multivariate canonical joint distribution" (MCJD). 
\begin{definition}
\label{def:canonical}
Suppose a group sequential test with $K$ analyses yields parameter estimates $\hat{\bs{\theta}}^{(1)},\dots,\hat{\bs{\theta}}^{(K)}$, then the multivariate canonical joint distribution holds if
\begin{enumerate}
\item $(\hat{\bs{\theta}}^{(1)},\dots,\hat{\bs{\theta}}^{(K)})$ is multivariate normal
\item $\hat{\bs{\theta}}^{(k)} \sim N_p(\bs{\theta}, \Sigma^{(k)}), \; 1\leq k\leq K$
\item $Cov(\hat{\bs{\theta}}^{(k_1)},\hat{\bs{\theta}}^{(k_2)}) = \Sigma^{(k_2)}, \; 1\leq k_1 \leq k_2 \leq K$.
\end{enumerate}
\end{definition}
Many estimators, for an array of datatypes, have the MCJD. \citep{jennison1997group} prove that the MCJD holds when the analysis is based on the maximum likelihood estimator, the maximum partial likelihood estimator for survival data and the parameter estimates from a normal linear model. This final case covers our sequence of estimates $\hat{\bs{\theta}}^{(1)},\dots,\hat{\bs{\theta}}^{(K)}$ in Equation~\eqref{eq:example_est}. The authors' results are commonly applied in the univariate setting because knowledge that the univariate canonical joint distribution (UCJD) holds makes light work of calculating boundary constants. Although univariate methods are more frequently used, there is no additional work required to prove that the MCJD holds for each of these broad-reaching datatypes. In Section~\ref{subsec:delta_method}, we prove that the UCJD holds approximately for the sequence of global statistics $\Delta(\hat{\bs{\theta}}^{(1)}),\dots,\Delta(\hat{\bs{\theta}}^{(K)})$. Therefore, the library of standard group sequential designs, (see \citet{pocock1977group}, \citet{o1979multiple} and \citet{gordon1983discrete}) can be directly applied.

Condition 3 of Definition~\ref{def:canonical} is known as the ``Markov property" or the ``independent increments property". The intuition here is that the distribution of the estimate at analysis $k$ depends on the previous analysis only and no earlier analyses. We can see this by noting that the distribution of $\hat{\bs{\theta}}^{(k)}$ given $\hat{\bs{\theta}}^{(k-1)}$ is such that 
\begin{equation}
\label{eq:cond2}
(\hat{\bs{\theta}}^{(k)}|\hat{\bs{\theta}}^{(k-1)}=\mathbf{x})\sim N_p(\bar{\bs{\theta}}^{(k)},\bar{\Sigma}^{(k)})
\end{equation}
where
\begin{align*}
\bar{\bs{\theta}}^{(k)}  &= \bs{\theta}+ \Sigma^{(k)}(\Sigma^{(k-1)})^{-1}(\mathbf{x}-\bs{\theta}) \\
\bar{\Sigma}^{(k)} &= \Sigma^{(k)} - \Sigma^{(k)}(\Sigma^{(k-1)})^{-1}\Sigma^{(k)}.
\end{align*}

\subsection{Error rates and boundary constants for the MGST}
\label{subsec:errors}
In what follows, we describe a process to evaluate the exact probability of the parameter estimates $\hat{\bs{\theta}}^{(1)},\dots,\hat{\bs{\theta}}^{(K)}$ being within a certain region whose boundaries are defined by the function $\Delta(\hat{\bs{\theta}}^{(k)})$. Then we are able to create multivariate group sequential tests and define stopping rules by finding boundary constants to control error rates. We begin by introducing the multivariate version of the recursive relation by~\citet{armitage1969repeated} to find sub-densities $g^{(k)}(\hat{\bs{\theta}}^{(k)};\bs{\theta}).$ This recursive relationship takes conditional distributions in a repetitive manner and makes use of the Markov property of the MCJD of Definition~\ref{def:canonical}. Denote $f(\hat{\bs{\theta}}^{(k)}|\hat{\bs{\theta}}^{(k-1)};\bs{\theta})$ as the conditional distribution of $\hat{\bs{\theta}}^{(k)}$ given $\hat{\bs{\theta}}^{(k-1)}$ in Equation~\eqref{eq:cond2}, then the subdensities are defined as
\begin{equation}
\label{eq:subdensities}
g^{(k)}(\hat{\bs{\theta}}^{(k)};\bs{\theta}) =
\begin{dcases}
f(\hat{\bs{\theta}}^{(1)};\bs{\theta}) & \text{ if } k=1 \\
\int_{\mathcal{C}^{(k-1)}}
g^{(k-1)}(\hat{\bs{\theta}}^{(k-1)};\bs{\theta})
f(\hat{\bs{\theta}}^{(k)}|\hat{\bs{\theta}}^{(k-1)};\bs{\theta})
d\hat{\bs{\theta}}^{(k-1)} & \text{ if } k=2,\dots,K.
\end{dcases}
\end{equation}
Now suppose that $\bs{\theta}_0\in \mathcal{N}$ is a vector of parameters in the null space which represent that the treatment does not work overall and that $\bs{\theta}_A\notin \mathcal{N}$ is vector of parameters representing an efficacious treatment. Let $\psi^{(k)}$ be the probability of stopping for efficacy and $\xi^{(k)}$ the probability of stopping for futility at analysis $k.$ These values are defined by
\begin{align}
\label{eq:prob_efficacy}
\psi^{(k)} &= \mathbb{P}_{\theta_0}\{\hat{\bs{\theta}}^{(1)} \in \mathcal{C}^{(1)},\dots,\hat{\bs{\theta}}^{(k-1)}\in\mathcal{C}^{(k-1)},\hat{\bs{\theta}}^{(k)}\in \mathcal{R}^{(k)}\} =
\int_{\mathcal{R}^{(k)}} g^{(k)}(\hat{\bs{\theta}}^{(k)};\bs{\theta}_0)
d\hat{\bs{\theta}}^{(k)} \\
\label{eq:prob_futility}
\xi^{(k)} &= \mathbb{P}_{\theta_A}\{\hat{\bs{\theta}}^{(1)} \in \mathcal{C}^{(1)},\dots,\hat{\bs{\theta}}^{(k-1)}\in\mathcal{C}^{(k-1)},\hat{\bs{\theta}}^{(k)}\in \mathcal{A}^{(k)}\} =
\int_{\mathcal{A}^{(k)}} g^{(k)}(\hat{\bs{\theta}}^{(k)};\bs{\theta}_A)
d\hat{\bs{\theta}}^{(k)}
\end{align}
The values of $\psi^{(k)}$ and $\xi^{(k)}$ for $k=1,\dots,K$ are prespecified at the design stage of the trial, which means that decision regions and therefore boundary constants $a^{(1)},\dots,a^{(K)}$ and $b^{(1)},\dots,b^{(K)}$, can be calculated to satisfy Equations~\eqref{eq:prob_efficacy} and~\eqref{eq:prob_futility}. There are a number of options for the configuration of $\psi^{(k)}$ and $\xi^{(k)}$. For example, the power family boundaries by~\citep{pampallona1994group} are a popular choice which result in more conservative decision making at early interim analyses. Many options suffer from reliance on an assumption of known variance of the test statistic at the design stage. This problem is known to be heightened for summary statistics such as the partial score function for survival data where variance is a function of the observed data. To combat this, the error-spending approach by~\citet{gordon1983discrete} was proposed as a flexible way of monitoring GSTs as it allows for using variance estimates to determine the configuration of the stage-wise error rates. Let $\mathcal{I}^{(1)},\dots,\mathcal{I}^{(K)}$ be information levels at analyses $1,\dots,K$ and $\mathcal{I}_{max}$ the anticipated maximum information at the end of the study, the choice for which is an essential feature of this design and will be discussed in Section~\ref{sec:sample_size}. The method partitions error rates according to
\begin{equation}
\label{eq:error_rates}
\begin{split}
\psi^{(1)}&=\pi_1(\mathcal{I}^{(1)}/\mathcal{I}_{max}), \hspace{1cm}
\psi^{(k)}=\pi_1(\mathcal{I}^{(k)}/\mathcal{I}_{max}) -\pi_1(\mathcal{I}^{(k-1)}/\mathcal{I}_{max})\text{ for }k=2,\dots,K\\
\xi^{(1)}&=\pi_2(\mathcal{I}^{(1)}/\mathcal{I}_{max}),\hspace{1cm}
\xi^{(k)}=\pi_2(\mathcal{I}^{(k)}/\mathcal{I}_{max}) -\pi_2(\mathcal{I}^{(k-1)}/\mathcal{I}_{max})\text{ for }k=2,\dots,K
\end{split}
\end{equation}
where $\pi_1(t)$ and $\pi_2(t)$ are non-decreasing functions and satisfy $\pi_1(0)=\pi_2(0)=0$ and $\pi_1(t)=\alpha$ for $t \geq 1$ and $\pi_2(t)=\beta$ for $t\geq 1.$ Throughout this paper, our examples will be under design considerations given by $\alpha=0.025,\beta=0.1$, $\pi_1(t)=\max\{\alpha t^2,\alpha\}$ and $\pi_2(t)=\max\{\beta t^2,\beta\}$. With $K=5$ and equally spaced information levels, this results in planned stage-wise type 1 error rates $\psi^{(1)},\dots,\psi^{(5)}=0.001,0.003,0.005,0.007,0.009$ and stage-wise type 2 error rates $\xi^{(1)},\dots,\xi^{(K)}=0.004,0.012,0.02,0.028,0.036.$

We now present a multivariate version of the error spending approach. As previously described, let $\hat{\Sigma}^{(k)}$ be an estimate of the $(p\times p)$ variance-covariance matrix $\Sigma^{(k)}$ at analysis $k$. A novel feature of this work is that we define information levels for global summary statistics by
\begin{equation}
\label{eq:information}
\mathcal{I}^{(k)}=|\hat{\Sigma}^{(k)}|^{-\frac{1}{2}}
\hspace{1cm}\text{for }k=1,\dots,K.
\end{equation}
This structure, which uses the determinant, is suitable because it respects proportionality to sample size and does not discriminate between positive and negative correlation structures i.e. a $2\times 2$ variance-covariance matrix with correlation $\rho=1$ has equal information to the corresponding matrix with $\rho=-1$.

When multiple treatments are concerned, it is necessary to ensure that type 1 error rates are controlled for all elements in the null space, $\mathcal{N}$, not only at a point value. Therefore, we emulate a feature of the approximate likelihood ratio test in that we choose $\bs{\theta}_0$ to be the point estimate of the treatment parameter which maximises the type 1 error rate and hence, there is no assumption needed regarding the direction of the alternative hypothesis. Formally, if $\sum_{k=1}^K\psi^{(k)}$ is the overall type 1 error rate under some value $\bs{\theta}\in\mathcal{N},$ then choose
$$\bs{\theta}_0=\text{arg}\,\max\limits_{\bs{\theta}\in \mathcal{N}}\left\{\sum_{k=1}^K\psi^{(k)}\right\}.$$
This choice is often visually obvious and/or simple to compute. In Figure~\ref{fig:reject_regions} we see that, out of all elements in $\mathcal{N}$, the point $\bs{\theta}_0=(0,0)$ has the shortest distance to $\mathcal{R}^{(1)}$ and naturally maximizes the type 1 error rate. The concept of choosing $\bs{\theta}_0$ is analogous to strong control of the family wise error rate in multiple testing procedures discussed by~\citep{proschan2020primer}.
\section{Methods for calculating probabilities and constructing boundary constants of multivariate tests}
\label{sec:methods}
\subsection{Multivariate Simpson's Rule}
\label{subsec:numerical_integration}

We have so far described a method to determine the exact probability of a sequence of parameter estimates lying with a sequence of decision regions. Unfortunately, the integral in question has no closed form analytic solution and methods to approximate this integral must therefore be employed. We consider three approximation methods including; numerical integration, the Delta method and Monte Carlo simulation. The numerical integration method introduces some novel ideas which are specific to the multivariate regions methodology and hence, we shall describe this method in detail. For the Delta and Monte Carlo methods, standard techniques are used and we refer the reader to the supplementary materials for details on the implementation. Each method is accompanied by a different source of error and we shall investigate the accuracy of each method in Section~\ref{sec:results}.

For the first method, we shall iterate Simpson's numerical integration rule over the dimensions of the vector $\hat{\bs{\theta}}^{(k)}$. Simpson's rule can be used to approximate the univariate integral of a general function $h(\hat{\theta})$ over the region $[l,u]$ by
\begin{equation}
\label{eq:simspon}
\int_l^u h(\hat{\theta})d\hat{\theta} \approx \sum_{i=1}^m w(i)h(x(i)),
\end{equation}
where $w(i)$ are weights and $x(i)$ are nodes for $i=1,\dots,m$. For our case, we shall see that the function $h(\theta)$ will be replaced by subdensities from Equation~\eqref{eq:subdensities} and the limits of integration $[l,u]$ will be explicitly defined based on the decision region $\mathcal{A}^{(k)}\mathcal{R}^{(k)}$ or $\mathcal{C}^{(k)}$, a novel feature of this work.

We describe the process of calculating the limits of integration in terms of the continuation region $\mathcal{C}^{(k)}= \{\hat{\bs{\theta}}^{(k)} : a^{(k)} \leq \Delta(\hat{\bs{\theta}}^{(k)}) < b^{(k)}\}$ and then the methodology for $\mathcal{A}^{(k)}$ and $\mathcal{R}^{(k)}$ follows analogously. The fact that $\Delta(\hat{\bs{\theta}}^{(k)})$ is a scalar means that, without loss of generality, the calculation will focus on dimension $p$. Let $\Delta^{-1}(\hat{\theta}_p^{(k)}|\hat{\theta}_1^{(k)},\dots,\hat{\theta}_{p-1}^{(k)})$ be the inverse function of $\Delta(\hat{\bs{\theta}}^{(k)})$ given elements $\hat{\theta}_1^{(k)},\dots,\hat{\theta}_{p-1}^{(k)}$ such that if $\Delta(\hat{\bs{\theta}}^{(k)})=z$ then $\Delta^{-1}(z|\hat{\theta}_1^{(k)},\dots,\hat{\theta}_{p-1}^{(k)})=\hat{\theta}_p^{(k)}.$ The following condition is needed for practical application, but not necessary for theory to hold.
\begin{condition}
\label{cond:inv_length}
The length of the vector $\Delta^{-1}(\hat{\theta}_p^{(k)}|\hat{\theta}_1^{(k)},\dots,\hat{\theta}_{p-1}^{(k)})$ is known for all $\hat{\bs{\theta}}^{(k)}\in \mathbb{R}^p$.
\end{condition}

It turns out to be surprisingly difficult to construct test statistics for which Condition~\ref{cond:inv_length} does not hold. Consider the test statistic given by $\Delta(\bs{\theta})=\sin(\theta_1\theta_2).$ The inverse function is therefore $\Delta^{-1}(\theta_2;\theta_1)=\sin^{-1}(\theta_2)/\theta_1$ which has an infinite number of values. Although it is unlikely that a test statistic of this nature would ever be considered, we therefore advise against using cyclic functions in the construction of the test statistic. When the analytic solution for $\Delta^{-1}(\hat{\theta}_p^{(k)}|\hat{\theta}_1^{(k)},\dots,\hat{\theta}_{p-1}^{(k)})$ does not exist, a search algorithm can be performed. In such a case, Condition~\ref{cond:inv_length} is needed to know how many times to perform the search to result in a vector of the desired length. In general, $\Delta(\hat{\bs{\theta}}^{(k)})$ need not be surjective or increasing which means that $\Delta^{-1}(z|\hat{\theta}_1^{(k)},\dots,\hat{\theta}_{p-1}^{(k)})$ has length zero or greater and we do not need to know the gradient of $\Delta(\hat{\bs{\theta}}^{(k)})$ at any point. Hence, this method is flexible and suitable for use in a wide range of trials and scenarios. 

We shall split the real line into segments defined by $\mathcal{C}^{(k)}.$ Suppose that
$$
\mathbf{y}=
(-\infty,
\Delta^{-1}(a^{(k)}|\hat{\theta}_1^{(k)},\dots,\hat{\theta}_{p-1}^{(k)})^T,
\Delta^{-1}(b^{(k)}|\hat{\theta}_1^{(k)},\dots,\hat{\theta}_{p-1}^{(k)})^T,
\infty)^T
$$
is a vector of length $q$. Then let
$
y_{(1)} \leq y_{(2)}\leq \dots \leq y_{(q)}
$
be the ordered elements of $\mathbf{y}$. By the ordering property, it can clearly be seen that the regions $[y_{(1)},y_{(2)}], \dots, [y_{(q-1)},y_{(q)}]$ are disjoint and their union is equal to $(-\infty,\infty)$. All elements $\hat{\bs{\theta}}^{(k)}$ such that $\hat{\theta}_p^{(k)} \in [y_{(s)},y_{(s+1)}],$ will have the same membership to the region $\mathcal{C}^{(k)}$. Therefore, denote the vector concatenating $\hat{\theta}_1,\dots,\hat{\theta}_{p-1}$ and the midpoint of $y_{(s)}$ and $y_{(s+1)}$ by
$$
M_s = (\hat{\theta}_1, \dots, \hat{\theta}_{p-1},\tfrac{1}{2}[y_{(s)}+y_{(s+1)}])^T,
$$
then for any general function $h(\hat{\bs{\theta}}^{(k)})$ integration over the region $\mathcal{C}^{(k)}$ is equal to
\begin{equation}
\label{eq:midpoint}
\int_{\mathcal{C}^{(k)}}h(\hat{\bs{\theta}}^{(k)})d\hat{\bs{\theta}}^{(k)}
=
\int_{-\infty}^{\infty}\dots\int_{-\infty}^{\infty}
\sum_{s=1}^{q-1} 
\mathbbm{1}_{\mathcal{C}^{(k)}}(\hat{\bs{\theta}}^{(k)}) 
\int_{y_{(s)}}^{y_{(s+1)}}
h(\hat{\bs{\theta}}^{(k)})
d\hat{\theta}_p^{(k)}\dots d\hat{\theta}_1^{(k)}
\end{equation}
where $\mathbbm{1}_{\mathcal{C}^{(k)}}(\hat{\bs{\theta}}^{(k)})$ denotes the indicator function that $\hat{\bs{\theta}}^{(k)}$ lies within $\mathcal{C}^{(k)}.$ The fact that the midpoint function can be taken outside of the integral means that the calculation only needs to be performed once for each element in the sum. This substantially speeds up calculation of the integral especially if evaluation of $\Delta(\hat{\bs{\theta}})$ is computationally expensive.

Combining the ideas of Equations~\eqref{eq:simspon} and~\eqref{eq:midpoint}, the subdensities in Equation~\eqref{eq:subdensities} can be approximated by
$$
\tilde{g}^{(1)}(x(i_1^{(1)}),\dots,x(i_p^{(1)});\bs{\theta})
=
f(x(i_1^{(1)}),\dots,x(i_p^{(1)}));\bs{\theta})
$$
and
\begin{equation*}
\begin{split}
\tilde{g}^{(k)}(x(i_1^{(k)}),\dots,x(i_p^{(k)});\bs{\theta})
=
\sum_{i_1^{(k-1)}=1}^{m_1^{(k-1)}} \dots
\sum_{i_{p-1}^{(k-1)}=1}^{m_{p-1}^{(k-1)}}
\sum_{s=1}^{q-1} 
\mathbbm{1}_{\mathcal{C}^{(k)}}(M_s) 
\sum_{i_s^{(k-1)}} w(i_1^{(k-1)})\dots w(i_{p-1}^{(k-1)}) \\
w(i_{s}^{(k-1)})\tilde{g}^{(k-1)}(x(i_1^{(k-1)}),\dots,x(i_{p-1}^{(k-1)});\bs{\theta}) 
f(x(i_1^{(k)}),\dots,x(i_{p-1}^{(k)})|x(i_1^{(k-1)}),\dots,x(i_{p-1}^{(k-1)});\bs{\theta})
\end{split}
\end{equation*}
for $k=2,\dots,K$. The choice of gridpoints $x(i_j^{(k)})$ and weights $w(i_j^{(k)})$ for $j=1,\dots,p$ and $k=1,\dots,K$ is discussed in the supplementary materials and uses standard results for Simpson's integration rule in one dimension. An important feature is that the number of weights and grid points is determined by a parameter $r$ and the choice for this value will be investigated in Section~\ref{sec:results}. A similar method is used to approximate Equations~\eqref{eq:prob_efficacy} and~\eqref{eq:prob_futility}. The difference here is that limits of integration are defined by regions $\mathcal{A}^{(k)}$ and $\mathcal{R}^{(k)}$ to replace $\mathcal{C}^{(k)}.$ This can be accomplished by replacing $(a^{(k)},b^{(k)})$ in the definitions of $\mathbf{y},q$ and $M_s$ with $(-\infty, a^{(k)})$ for $\mathcal{A}^{(k)}$ and with $(b^{(k)},\infty)$ for $\mathcal{R}^{(k)}$.  We note that $q\geq 2$ and when Simpson's rule is applied, the locations of the nodes are given finite values, so it is not a problem to calculate the midpoint of a line segment with non-finite bound(s). The stage-wise type 1 and type 2 error rates can therefore be approximated by
\begin{align*}
\psi^{(k)}
\approx
\sum_{i_1^{(k)}=1}^{m_1^{(k)}} \dots
\sum_{i_{p-1}^{(k)}=1}^{m_{p-1}^{(k)}}
\sum_{s=1}^{q-1} 
\mathbbm{1}_{\mathcal{R}^{(k)}}(M_s) 
\sum_{i_s^{(k-1)}} w(i_1^{(k)})\dots w(i_{p-1}^{(k)}) w(i_{s}^{(k)})
\tilde{g}^{(k)}(x(i_1^{(k)}),\dots,x(i_p^{(k)});\bs{\theta}_0) \\
\xi^{(k)}
\approx
\sum_{i_1^{(k)}=1}^{m_1^{(k)}} \dots
\sum_{i_{p-1}^{(k)}=1}^{m_{p-1}^{(k)}}
\sum_{s=1}^{q-1} 
\mathbbm{1}_{\mathcal{A}^{(k)}}(M_s) 
\sum_{i_s^{(k-1)}} w(i_1^{(k)})\dots w(i_{p-1}^{(k)}) w(i_{s}^{(k)})
\tilde{g}^{(k)}(x(i_1^{(k)}),\dots,x(i_p^{(k)});\bs{\theta}_A)
\end{align*}

With the methodology in place, we have presented a multidimensional numerical integration rule to evaluate the probabilities in Equations~\eqref{eq:prob_efficacy} and~\eqref{eq:prob_futility}. At the design stage, $\psi^{(k)}$ and $\xi^{(k)}$ are pre-specified for $k=1,\dots,K$ and the aim is then to find the boundary constants $a^{(k)}$ and $b^{(k)}$ to satisfy these constraints. This can be achieved by employing a root-finding algorithm. There are three observations to make here; firstly, in the same way as the usual group sequential test for univariate statistics, the nodes, weights and sub-densities can be stored at analysis $k$. Secondly, the unbounded gridpoints along with their weights and sub-densities do not change with $a^{(k)}$ and $b^{(k)}$, only the values at the edges must be recalculated when implementing the root-finding. Finally, efficiencies may be gained by considering the order of the dimensions at the onset. For example, it is easier to find the boundary constants for the function $\Delta(\theta_1,\theta_2)=\theta_1^2\theta_2$ compared with $\Delta(\theta_1,\theta_2)=\theta_1\theta_2^2.$

\subsection{Delta approximation}
\label{subsec:delta_method}
We can estimate probability functions and therefore boundary constants by utilizing another approximation at the distributional level. Using the Delta method by~\citep{doob1935limiting}, we can show approximately that the following distribution holds.
\begin{theorem}
\label{theorem:delta}
Let $\bs{\theta}$ be a $p\times 1$  vector of parameters in a statistical model. Suppose that $\hat{\bs{\theta}}^{(k)}$ is the parameter estimate for $\bs{\theta}$ found at analysis $k$ of a group sequential trial with $K$ analyses and that $\Sigma^{(k)}$ is the variance-covariance matrix for $\hat{\bs{\theta}}^{(k)}$. Let $\Delta(\bs{\theta})$ be a function which returns a scalar output and suppose that the sequence $\hat{\bs{\theta}}^{(1)}, \dots, \hat{\bs{\theta}}^{(K)}$ has the MCJD. Then the sequence of estimates $\Delta(\hat{\bs{\theta}}^{(1)}), \dots, \Delta(\hat{\bs{\theta}}^{(K)})$ are such that
\begin{enumerate}
\item $\Delta(\hat{\bs{\theta}}^{(1)}), \dots, \Delta(\hat{\bs{\theta}}^{(K)})$ is multivariate normally distributed
\item $\Delta(\hat{\bs{\theta}}^{(k)}) \sim 
N\left( \Delta(\bs{\theta}), \left[\frac{\partial\Delta(\bs{\theta})}{\partial \bs{\theta}}\right]^T \Sigma^{(k)} \left[\frac{\partial\Delta(\bs{\theta})}{\partial \bs{\theta}}\right]\right)$ for $1\leq k\leq K$
\item $Cov\left(\Delta(\hat{\bs{\theta}}^{(k_1)}),\Delta(\hat{\bs{\theta}}^{(k_2)})\right) = \left[\frac{\partial\Delta(\bs{\theta})}{\partial \bs{\theta}}\right]^T\Sigma^{(k_2)}\left[\frac{\partial\Delta(\bs{\theta})}{\partial \bs{\theta}}\right]$ for $k=1\leq k_1 \leq k_2 \leq K.$
\end{enumerate}
\end{theorem}

\begin{proof}
See the supplementary materials.
\end{proof}
Theorem~\ref{theorem:delta} is equivalent to saying that the sequence of estimates $\Delta(\hat{\bs{\theta}}^{(1)}),\dots, \Delta(\hat{\bs{\theta}}^{(K)})$ has approximately the UCJD. The subdensities and probabilities in Equations~\eqref{eq:prob_efficacy} and~\eqref{eq:prob_futility} can be approximated using this distribution and the resulting integrals can be evaluated using a univariate version of the numerical integration of Section~\ref{subsec:numerical_integration}. Each sub-density takes a scalar input value and this becomes equivalent to the methods in Chapter 19 by~\citet{jennison2000group}. Due to the scalar nature of the output $\Delta(\hat{\bs{\theta}}^{(k)})$ and the Markov property of the UCJD, integration will be at most in two dimensions which is much more efficient than the $2p$ dimensional integration of Section~\ref{subsec:numerical_integration}. In a similar manner to Section~\ref{subsec:numerical_integration}, the boundary constants $a^{(1)},\dots,a^{(K)}$ and $b^{(1)},\dots,b^{(K)}$ are calculated using a root finding algorithm which repeatedly evaluates Equations~\eqref{eq:prob_efficacy} and~\eqref{eq:prob_futility} under the predefined significance and power requirements.

The Delta method is based on the first order Taylor series expansion of the function $\Delta(\hat{\bs{\theta}}^{(k)})$ around at the point $\bs{\theta}$. In some cases, the error in the distributional approximation is too big to ensure validity of the method. This situation occurs for studies with small sample sizes where there is high variability in $\hat{\bs{\theta}}^{(k)}$ so that the estimate is unlikely to be close to the tangent line. In some situations this issue cannot always be resolved by increasing the sample size. This occurs when the higher order terms do not converge faster to zero than the first order term. To demonstrate this, we consider the example function $\Delta(\bs{\theta})$  in Equation~\eqref{eq:example}. For this function, the Taylor expansion is given by
\begin{equation*}
\hat{\theta}_1^{(k)}\hat{\theta}_2^{(k)}
=
(\theta_1\theta_2) \;+ \;
(\hat{\theta}_1^{(k)}\theta_2 + \theta_1\hat{\theta}_2^{(k)} -2\theta_1\theta_2)  \;+ \;
(\hat{\theta}_1^{(k)}\hat{\theta}_2^{(k)}-\hat{\theta}_1^{(k)}\theta_2-\theta_1\hat{\theta}_2^{(k)}+\theta_1\theta_2)
\end{equation*}
and higher order terms are equal to zero. The first and second order elements of the Taylor expansion contain terms of the form $\hat{\theta}_{j_1}^{(k)}\theta_{j_2}$ which have variance of the same magnitude. Hence, it is detrimental to make the assumption that only the first order term is needed but not the second. In the supplementary materials, we provide evidence that the estimates obtained from the Delta method are not normally distributed for this global summary statistic. The problem is particularly relevant in the tails and due to the nature of trial design, this may lead to poor control of error rates.

\subsection{Monte Carlo}
\label{subsec:monte_carlo}
The final method for calculating probability functions and boundary constants is the simplest and easiest to implement. This is a Monte Carlo simulation based technique (see \citep{kroese2014monte}) where we repeatedly sample the sequence $\hat{\bs{\theta}}^{(1)},\dots,\hat{\bs{\theta}}^{(K)}$ from the MCJD and estimate probabilities as the proportion of samples which satisfy the desired constraints. Further, we can calculate quantiles and boundary constants in the usual manner for non-parametric statistics. Details for implementing this standard method are provided in the supplementary materials. In Section~\ref{sec:results} we assess the results as we vary the number of samples, denoted by $N$, and we note that there will be a total of $2N$ replicates since we simulate under $H_0$ and $H_A$. One feature worth noting is that the error in this method increases with analysis since the number of replicates in the continuation region reduces as the trial progresses. As the error propagates throughout the trial, this causes worsening estimation of the boundary constants at later analyses. We show this through an example in Section~\ref{sec:results}.
\section{Sample size calculation}
\label{sec:sample_size}
An important aspect of any clinical trial design, whether fixed or group sequential, is the sample size calculation. We have presented theory for calculating boundary constants and probability functions and three methods for the practical implementation of this theory which we shall make use of to accurately determine the sample size in an efficient manner.

Throughout this section, we shall discuss sample size calculation in terms of information levels to facilitate the error spending approach. Such information levels were introduced during the discussion of error spending tests and are given by $\mathcal{I}^{(k)}=|\hat{\Sigma}^{(k)}|^{-1/2}$. For the purpose of sample size calculations, we shall assume that the variance-covariance matrices take the form $\Sigma^{(k)}=M/n^{(k)}$ where $M$ is a known matrix of dimension $(p\times p)$ and $n^{(k)}$ is the cumulative sample size per treatment arm at analysis $k$. This implies the relationship
\begin{equation}
\label{eq:sample_size}
\mathcal{I}^{(k)}=n^{(k)}|M|^{-\frac{1}{2}}.
\end{equation}
Further, if the sample sizes at each analysis are of equal sizes in the MGST, so that $n^{(k)} \propto k,$ we will therefore have $\mathcal{I}^{(k)}/\mathcal{I}^{(K)}=k/K.$
For the remainder of the sample size calculation, we shall utilize these two assumptions. We note that this is equivalent to the necessary assumptions in the sample size calculation for fixed and group sequential tests in one dimension. In which case, the value of a nuisance parameter (usually denoted $\sigma^2$) is specified and we also assume equally spaced information levels. To emphasize this point, the assumptions are only required for the sample size calculation but do not affect the analysis of the collected data. In the supplementary materials, we perform sensitivity analyses which investigate the effects of misspecifiation of the nuisance matrix $M.$

In general, a sample size calculation simply finds the root of a probability equation with constraints determined by significance level and power requirements. Denote $\mathcal{I}_{fix}$ as the required information in the fixed sample trial, this takes the value of $\mathcal{I}^{(1)}$ that satisfies
\begin{equation*}
\begin{split}
\alpha
&= \mathbb{P}_{\bs{\theta}_0}\{\hat{\bs{\theta}}^{(1)}\in \mathcal{R}^{(1)}\}\\
\beta 
&= \mathbb{P}_{\bs{\theta}_A}\{\hat{\bs{\theta}}^{(1)}\in \mathcal{A}^{(1)}\}\\
a^{(1)}&=b^{(1)}.
\end{split}
\end{equation*}
As discussed, there is no closed form solution for these probabilities and boundary constants so we must employ one of the approximations of Section~\ref{sec:methods} during the root finding process. For reasons discussed in Section~\ref{sec:results}, here we use Simpson's numerical integration method with $r=128$.

We now find the maximum information level required for a MGST using the error spending approach. The constraints for this part of the sample size calculation are given by the planned stage-wise type 1 error rates $\psi^{(1)}$ and stage-wise type 2 error rates $\xi^{(k)}$ for $k=1,\dots,K$ which are determined using the error spending functions and Equation~\eqref{eq:error_rates}. The sample size calculation for other choices of boundary configurations is equivalent to fixing the values of $\psi^{(k)}$ and $\xi^{(k)}$. Our aim is to find $\mathcal{I}_{max}$ as the value of $\mathcal{I}^{(K)}$ subject to constraints
\begin{equation*}
\begin{split}
\psi^{(k)} 
&=
\mathbb{P}_{\bs{\theta}_0}\{
\hat{\bs{\theta}}^{(1)}\in \mathcal{C}^{(1)}, \dots,
\hat{\bs{\theta}}^{(k-1)}\in \mathcal{C}^{(k-1)},
\hat{\bs{\theta}}^{(k)}\in \mathcal{R}^{(k)}\}
\hspace{1cm}\text{for }k=1,\dots,K\\
\xi^{(k)} 
&=
\mathbb{P}_{\bs{\theta}_A}\{\hat{\bs{\theta}}^{(1)}\in \mathcal{C}^{(1)}, \dots,
\hat{\bs{\theta}}^{(k-1)}\in \mathcal{C}^{(k-1)},
\hat{\bs{\theta}}^{(k)}\in \mathcal{A}^{(k)}\}
\hspace{1cm}\text{for }k=1,\dots,K\\
a^{(K)}&=b^{(K)}.
\end{split}
\end{equation*}
Again, one of the methods of Section~\ref{sec:methods} must be used to compute the integral and perform the search for $\mathcal{I}_{max}$. The sensible choice would be to use Simpson's numerical integration rule. The challenge here is that in $p$ dimensions and $K$ analyses, the numerical integration method is computationally expensive and hence, repetitively evaluating this integral in a root finding manner is infeasible. Instead, let us suppose that we know the value of the ratio $R=\mathcal{I}_{max}/\mathcal{I}_{fix}$, then we would be able to find $\mathcal{I}_{max}=R\mathcal{I}_{fix}$ without calculating the group sequential probabilities with Simpson's rule. It turns out that by utilizing the Delta method, even though individual estimates for $\mathcal{I}_{fix}$ and $\mathcal{I}_{max}$ are inaccurate, we can find an accurate estimate for the ratio. The intuition here is that the error in the distributional assumption of the Delta method is proportional to sample size, hence in the ratio, the effect is canceled out. The results in Section~\ref{sec:results} back up this finding.  The one dimensional nature of the Delta method allows for computational efficiency and in turn makes calculation of $R=\tilde{\mathcal{I}}_{max}/\tilde{\mathcal{I}}_{fix}$ possible. The tilde notation is here used to denote that the Delta method with $r=128$ is used in the calculation of such values. 

We consider an example with global summary statistic function given by $\Delta(\theta_1,\theta_2)$ in Equation~\eqref{eq:example} and parameter values given by
\begin{equation}
\label{eq:param_vals}
\bs{\theta}_0=(0,0), \bs{\theta}_A = (1.625,1.625), \Sigma^{(k)} =\frac{1}{n^{(k)}} \left(\begin{array}{cc} 40 & 10 \\ 10 & 40 \end{array}\right).
\end{equation}
The fixed sample trial is designed with significance level $\alpha=0.025$ at the point $\bs{\theta}_0$ and power $1-\beta=0.9$ at the point $\bs{\theta}_A$. Using Simpson's rule with $r=128$ we find $\mathcal{I}_{fix}=2.653$ and using the Delta method with $r=128$ we obtain $\tilde{\mathcal{I}}_{fix}=10.274.$ The MGST is designed with a total of $K=5$ analyses and error spending functions $\pi_1(t)=\max\{\alpha t^2,\alpha\}$ and $\pi_2(t)=\max\{\beta t^2,\beta\}$. Using the Delta method with $r=128,$ we find a value of $\tilde{\mathcal{I}}_{max}=11.305$ to satisfy these constraints. Hence, the true maximum required information level is given by $\mathcal{I}_{max}=2.920$. Finally, we can relate information levels and sample size using Equation~\eqref{eq:sample_size}. Hence, the total sample sizes required for the fixed sample trial and MGST respectively are $n_{fix}=102.769$ and $n_{gst}=113.082$. Rounding up so that we have integer values of patients per treatment arm at each stage, we design the fixed sample trial with $n^{(1)}=103$ per treatment arm and the MGST with $n^{(k)}=23k$ per treatment arm. In Section~\ref{sec:results} we show that these sample sizes are accurate resulting in slightly higher power than required due to the discrete nature of sample size.

\section{Comparison of methods}
\label{sec:results}

We have presented three methods to calculate the boundary constants $a^{(k)}$ and $b^{(k)}$ for $k=1,\dots,K$ to construct both fixed sample and group sequential tests. In this section we compare the accuracy of each method for two choices of the global statistic function which are given by $\Delta(\bs{\theta})=\theta_1+\theta_2$ and $\Delta(\bs{\theta})$ in Equation~\eqref{eq:example}. We also consider properties of each method such as the number of gridpoints $r$ for Simpson's rule and the number of simulations $N$ for the Monte Carlo method and hence compare the efficiency of the methods. Further, we separate comparisons between fixed and group sequential tests as the computational burden increases notably when we extend the trial beyond $K=1$ analyses. Throughout this section, we shall assume that the fixed sample trial is designed with significance level $\alpha=0.025$ and power $1-\beta=0.9$. The group sequential test has a total of $K=5$ analyses and we shall assume that the variance matrices $\Sigma^{(1)},\dots,\Sigma^{(K)}$ are known without error so that informative evaluations can be made. Hence, the trial is designed with stage-wise type 1 error rates $\psi^{(1)},\dots,\psi^{(5)}=0.001,0.003,0.005,0.007,0.009$ and stage-wise type 2 error rates $\xi^{(1)},\dots,\xi^{(K)}=0.004,0.012,0.02,0.028,0.036$ which represents more conservative decision making at the earlier analyses.

\begin{table}
  \small
  \begin{center}
  \begin{tabular}{llccclccc}
  &\multicolumn{4}{c}{$\Delta(\bs{\theta})=\theta_1+\theta_2$}&
  \multicolumn{4}{c}{$\Delta(\bs{\theta})=\theta_1\theta_2$}\\
  Method &
  Property &$b^{(1)}$&$\psi^{(1)}$ & $\xi^{(1)}$&
  Property &$b^{(1)}$&$\psi^{(1)}$ & $\xi^{(1)}$\\
  \rule{0pt}{4ex}Simpson & $r=9$&1.9599&0.02501&0.09850&
    $r=9$&0.8234&0.02501&0.09935\\Simpson & $r=10$&1.9599&0.02500&0.09851&
    $r=10$&0.8234&0.02500&0.09936\\Simpson & $r=11$&1.9599&0.02500&0.09851&
    $r=11$&0.8234&0.02500&0.09936\\\rule{0pt}{4ex}Delta & $r=7$&1.9599&0.02500&0.09850&
    $r=9$&3.1382&0.00013&0.64632\\Delta & $r=8$&1.9599&0.02500&0.09851&
    $r=10$&3.1382&0.00013&0.64633\\Delta & $r=9$&1.9599&0.02500&0.09851&
    $r=11$&3.1382&0.00013&0.64633\\\rule{0pt}{4ex}Monte Carlo & $N=10^5$&1.9640&0.02476&0.09923&
    $N=10^5$&0.8245&0.02494&0.09955\\Monte Carlo & $N=10^6$&1.9578&0.02513&0.09814&
    $N=10^6$&0.8274&0.02476&0.10004\\Monte Carlo & $N=10^7$&1.9597&0.02501&0.09848&
    $N=10^7$&0.8232&0.02502&0.09932\\
  \end{tabular}
  \caption{Comparison of methods for constructing a fixed sample trial using the linear function
  $\Delta(\theta_1,\theta_2)=\theta_1+\theta_2$ and non-linear function 
  $\Delta(\theta_1,\theta_2)$ in Equation~\eqref{eq:example}.
  Simulations use true parameter values given in~\eqref{eq:param_vals} and a fixed
  sample size of $n^{(1)}=100$ for the linear case and $n^{(1)}=103$
  for the non-linear case.
  True observed probabilities $\psi^{(1)}$ and $\xi^{(1)}$ were calculated using
  the Delta method with r=128 for the linear case and Simpson's method with
  $r=32$ for the non-linear case.
  \label{tbl:fixed}}
  \end{center}
  \end{table}
Table~\ref{tbl:fixed} shows the performance of each method in a fixed sample trial using the linear function $\Delta(\bs{\theta})=\theta_1+\theta_2.$ For each method, we begin by calculating the boundary constant $b^{(1)}$. Then, we can find \emph{true} observed probabilities $\psi^{(1)}$ and $\xi^{(1)}$ in Equations~\eqref{eq:prob_efficacy} and~\eqref{eq:prob_futility} for this value of $b^{(1)}$. For calculating \emph{true} probabilities, we have used the Delta method with a very high value of $r=128$ for high accuracy. This is appropriate because the function $\Delta(\bs{\theta})$ is linear and so the Delta method approximation is exact in this case. We recognize that theoretical properties of this estimator are known and we have checked our values of $\psi^{(1)}$ and $\xi^{(1)}$ against the function ``pnorm" in R. Following the methods in Section~\ref{sec:sample_size} and the parameters given by~\eqref{eq:param_vals}, the required sample size for the fixed sample trial is $n^{(1)} = 100$ per treatment arm.  We note that $\psi^{(1)}$ is equal to the required significance level exactly however the observed type 2 error is slightly lower than planned due to the discrete value of the sample size. With $r=10$ controlling the number of grid points in each dimension, Simpson's method finds boundary constants that have error rates accurate to 5 decimal places. Unsurprisingly, the Delta Method works well in this scenario, converging to the true values of $\psi^{(1)}$ and $\xi^{(1)}$ at a faster rate than the multivariate Simpson's rule with only $r=8$ required. The Monte Carlo method appears to perform adequately in practice with all values close to the truth. However, due to randomness in this method, increasing the number of simulations $N$ to extreme values still does not guarantee convergence to the truth and some type 1 error rates are inflated which casts doubt on the accuracy of the Monte Carlo method for other functions $\Delta(\cdot).$ The number of unbounded grid points, without truncating for the boundary and including the midpoints, is equal to $(12r-1)^p$ for the multivariate regions method. With $r=16$ and $p=2$ we are required to calculate $36481$ values of the density function which is much more efficient than the Monte Carlo method even with $N=10^5$ which leads to poor accuracy. Results for other significance and power requirements are presented in the supplementary materials which show a similar pattern for all methods. However we see that slightly higher values of $r$ (and hence more grid points) are required for the numerical integration methods when considering values further into the tails of this distribution. For example when $\alpha=0.001$, we need $r=17$ for the multivariate Simpson's method and $r=12$ for the Delta method to reach the same level of accuracy as in Table~\ref{tbl:fixed}.

The next case considered was a fixed sample trial using the function $\Delta(\bs{\theta})$  in Equation~\eqref{eq:example}. With the sample size $n^{(1)}=103$ per treatment arm and the parameter values given in~\eqref{eq:param_vals}, Table~\ref{tbl:fixed} shows the results of this study. For this choice of $\Delta(\cdot)$, we have used Simpson's rule with a very high value of $r=128$ to act as the surrogate true calculation of $\psi^{(1)}$ and $\xi^{(1)}$ at each value of $b^{(1)}.$ We believe this is appropriate because of the accuracy of the method for the linear case where the truth was known. For Simpson's method and Monte Carlo methods, the results display similar patterns to the fixed sample case with the linear choice for $\Delta(\bs{\theta}).$ That is, both methods converge to the true value with the numerical integration method being more efficient in terms of number of required calculations. Compared to the linear case, slightly higher numbers of gridpoints are required for the same accuracy level which is likely due to the shape of the region which approaches an asymptote in the tails of the distribution. The Delta method does not appear to work in this context. These estimates have converged, to a reasonable level of accuracy, to biased values and are therefore inconsistent with the truth. Note that, we have checked that increasing the value of $r$ does not decrease the bias in estimation. Although the error in approximation may be reduced for large sample sizes, the bias in the Delta method and the calculation problems under $H_0$, casts doubt on the validity of the method overall. In Section~\ref{subsec:delta_method}, we showed a theoretical evidence for why the Delta method is not expected to work for functions that are non-linear. The recommendation for fixed sample trials with summary statistics that are non-linear is therefore clear; the multivariate Simpson's rule is the most accurate and efficient choice of method.

\begin{table}
  \small
  \begin{center}
  \begin{tabular}{lllcccccccc}
  &&&\multicolumn{4}{c}{$\Delta(\bs{\theta})=\theta_1+\theta_2$}&
  \multicolumn{4}{c}{$\Delta(\bs{\theta})=\theta_1\theta_2$}\\
  Method & Property & $k$ &
  $a^{(k)}$ & $b^{(k)}$ & $\psi^{(k)}$ & $\xi^{(k)}$&
  $a^{(k)}$ & $b^{(k)}$ & $\psi^{(k)}$ & $\xi^{(k)}$\\
  \rule{0pt}{4ex}Simpson&$r=6$&1&-2.4044&6.5816&0.00101&0.00400&-3.9221&9.8568&0.00100&0.00400\\Simpson&$r=6$&2&-0.0732&4.0915&0.00300&0.01200&-0.8151&3.7549&0.00300&0.01200\\Simpson&$r=6$&3&0.9136&3.0436&0.00500&0.02000&-0.0352&2.0504&0.00500&0.02000\\Simpson&$r=6$&4&1.5027&2.4260&0.00700&0.02800&0.3634&1.2838&0.00700&0.02800\\Simpson&$r=6$&5&1.9554&1.9554&0.00900&0.03492&0.8295&0.8295&0.00900&0.03642\\\rule{0pt}{4ex}Delta&$r=128$&1&-2.4042&6.5884&0.00100&0.00400&-6.3455&10.4708&0.00073&0.00052\\Delta&$r=128$&2&-0.0730&4.0917&0.00300&0.01200&-2.6406&6.5028&0.00016&0.00050\\Delta&$r=128$&3&0.9137&3.0435&0.00500&0.02000&-1.0727&4.8374&0.00006&0.00123\\Delta&$r=128$&4&1.5028&2.4259&0.00700&0.02800&-0.1424&3.8624&0.00004&0.00615\\Delta&$r=128$&5&1.9553&1.9553&0.00900&0.03490&3.2027&3.2027&0.00003&0.60063\\\rule{0pt}{4ex}Monte Carlo&$N=10^7$&1&-2.4106&6.5755&0.00102&0.00396&-3.9243&9.8869&0.00098&0.00399\\Monte Carlo&$N=10^7$&2&-0.1076&4.1550&0.00262&0.01127&-0.8475&3.7809&0.00292&0.01126\\Monte Carlo&$N=10^7$&3&0.7577&3.2954&0.00255&0.01407&-0.1397&2.1487&0.00420&0.01367\\Monte Carlo&$N=10^7$&4&1.1897&2.8880&0.00173&0.01289&0.0970&1.5175&0.00388&0.01092\\Monte Carlo&$N=10^7$&5&2.6720&2.6720&0.00096&0.18574&1.3758&1.3758&0.00153&0.12687\\
  \end{tabular}
  \caption{Comparison of methods for constructing group sequential trials using the linear function
  $\Delta(\theta_1,\theta_2)=\theta_1+\theta_2$ and the non-linear function
  $\Delta(\theta_1,\theta_2)$ in Equation~\eqref{eq:example}.
  Simulations use true parameter values given in~\eqref{eq:param_vals} and a
  sample size of $n^{(k)}=22k$ for the linear case and 
  $n^{(k)}=23k$ for the non-linear case.
  True observed probabilities $\psi^{(k)}$ and $\xi^{(k)}$ were calculated using
  the Delta method with $r=128$ for the linear case and Simpson's Method with $r=16$
  for the non-linear case.
  \label{tbl:gst}}
  \end{center}
  \end{table}
We now consider how the comparison of methods extends to a group sequential trial with a total of $K=5$ analyses. We return to the linear function $\Delta(\bs{\theta})=\theta_1+\theta_2$ so that methods can be checked against the known truth. Following the sample size calculation in Section~\ref{sec:sample_size} and true parameter values given by~\eqref{eq:param_vals}, a total information level of $\mathcal{I}_{max}=2.826$ is required which results in group sizes of $n^{(k)}=22k$ patients per treatment arm. Table~\ref{tbl:gst} shows the estimated boundary constants $a^{(k)}$ and $b^{(k)}$ for $k=1,\dots,K$ constructed under each method. For each set of boundary constants, the probabilities $\psi^{(k)}$ and $\xi^{(k)}$ for $k=1,\dots,K$ are then calculated using the Delta with $r=128$. Similarly to the fixed sample case, this is an appropriate calculation for the \emph{true} probabilities because the function is linear and we have also checked these values against the function ``pmvnorm" from the package \emph{mvtnorm} in R by~\citep{genz2020package}. Both numerical integration methods accurately identify the correct boundary constants and the observed Type 1 error rates are as required. The Monte Carlo method is inaccurate in this case especially for later analyses. This is because we start with a large $N=10^7$ but then for each analysis that passes, the number of observed simulations in the continuation region reduces. Calculation of the boundary constants, via the non-parametric quantile function of Section~\ref{subsec:monte_carlo}, is then based on fewer observations and this method is therefore subject to higher inaccuracy. In fact, across all analyses, the Monte Carlo method has a combined type 1 error rate of 0.009 and type 2 error rate 0.228 which are far from the planned $\alpha=0.025$ and $\beta=0.1.$ For Simpson's rule, a total of $(12r-1)^{2p}$ unbounded univariate nodes are produced for the boundary construction at analysis $k$. With $K=5,p=2$ and $r=6,$ this is approximately $2.5\times 10^7,$ a value comparable to $2N=2\times10^7$ (simulations under $H_0$ and $H_A$), hence these methods take roughly the same amount of time. Similar results for different configurations of the required type 1 and type 2 error rates are found in the supplementary materials. The results follow a similar pattern regardless of error configuration. 

Finally, Table~\ref{tbl:gst} shows the comparison of the methods for a MGST with $K=5$ and non-linear function $\Delta(\bs{\theta})$  in Equation~\eqref{eq:example}. Again, the boundary constants are calculated according to each method and the number of gridpoints in Simpson's method has been chosen to reflect a similar number of calculations (and therefore computation time) as the Monte Carlo method. Evaluation of the true probabilities $\psi^{(k)}$ and $\xi^{(k)}$ given each set of $a^{(k)}$ and $b^{(k)}$ for $k=1,\dots,K$ is more challenging in this setting since the only accurate method for calculation is the multivariate Simpson's rule with a large number of gridpoints. Ideally, we would set $r=128$, but with $p=2$ dimensions and $K=5$ analyses, computation time is in the timescale of weeks. However, we have shown that $r=6$ produces sufficient accuracy for the linear function and we have reason to believe that this extends to the non-linear case. Therefore setting $r=16$ for calculation of the true probabilities is sufficiently large to reliably compare these methods, and reduces computation time to under 4 hours. For this trial, the sample size is given by $n^{(k)}=23k$ as described in Section~\ref{sec:sample_size}. Simpson's rule found the correct boundary constants even with small $r=6$ and the total type 2 error rate was found to be 0.100 which confirms that accuracy of the sample size calculation. The Delta method is unsuccessful at determining boundary constants when the global summary function is not well approximated by a linear function. In this case, the total type 1 and 2 error rates were 0.001 and 0.609 respectively which suggest that this method is redundant. The Monte Carlo method works moderately well in this setting with total type 1 and 2 error rates close to expected. However the error in the boundary constant calculation increases as the trial progresses and could even construct a trial with inflated type 1 error rates because it is subject to randomness. For a completely robust and accurate trial design, the multivariate Simpson's rule is the best option for constructing MGSTs. 
\section{Discussion}
A novel aspect of this work is that only a single hypothesis is tested throughout the trial. The usual framework would consider a hypothesis test for each parameter and then use a multiple testing procedure to control error rates. The Bonferroni method is commonly used for multiplicity adjustment because of its simplicity. For a fixed sample clinical trial and parameter values given by~\eqref{eq:param_vals}, the Bonferroni method with equal weights would reject the null hypothesis on any endpoint such that the p-value is less than $\alpha/2=0.0125.$ With fixed sample size $n^{(1)}=121$ per treatment arm, this has power 0.55 to claim significance for both endpoints. Our method gave power 0.9 using the global summary statistic $\Delta(\bs{\theta})$ in Equation~\eqref{eq:example} and a smaller sample size gave power 0.9 using the linear global statistic. A formal comparison of MGST with multiple testing procedures for each dimension independently could motivate future research.

In one dimension, the final property of the CJD is often referred to as the independent increments property and makes implementation of GSTs straightforward. Further, this property allows the GST methodology to be extended to many other types of adaptive designs including, but not limited to, Seamless Phase II/III trials, adaptive enrichment designs and multi-arm multi-stage trials (see \citep{stallard2011seamless}, \citep{simon2013adaptive}, \citep{wason2012optimal}). Therefore, there are many opportunities to extend the multivariate regions methodology to other adaptive designs and take hold of the benefits that arise from incorporating data on multiple endpoints.

\section*{Software}
All statistical computing and analyses were performed using the software environment R version 4.2.3.
Software relating to the examples in this paper is available at https://github.com/abigailburdon/Decision-regions-for-multivariate-tests.

\section*{Acknowledgements}
This project has received funding from the European Union’s Horizon 2020 research and
innovation programme under grant agreement No 965397. TJ also received funding from the
UK Medical Research Council (MC\_UU\_00002/14). For the purpose of open access, the
author has applied a Creative Commons Attribution (CC BY) licence to any Author Accepted
Manuscript version arising. 

\bibliographystyle{plainnat}  
\bibliography{enrichment}

\end{document}


\maketitle
\section{Details for implementing Simpson's numerical integration rule in multiple dimensions}
We provide details for implementing the methods of Section~3.1 in the main text which is the multivariate version of Simspson's numerical integration rule. In particular we describe the choice of gridpoints and weights. Following Chapter 19 of~\citet{jennison2000group}, we create a set of unbounded gridpoints before truncating these points according to the limits of integration. In our case, we are interested in normally distributed data so it is a sensible choice to concentrate grid points in the center of the distribution and place fewer nodes in the rapidly decreasing tails. In one dimension, the middle two thirds are placed within three standard distributions of the mean, hence for some integer value of $r$, we define elements
$$
\phi_j^{(k)}(i) = \begin{cases}
\theta_j^{(k)}+\sqrt{\Sigma_{jj}^{(k)}}(-3-4\log(r/i) & i=1,\dots,r-1\\
\theta_j^{(k)}+\sqrt{\Sigma_{jj}^{(k)}}(-3+3(i-r)/2r & i=r,\dots,5r\\
\theta_j^{(k)}+\sqrt{\Sigma_{jj}^{(k)}}(-3+4\log(r/i) & i=5r+1,\dots,6r-1\\
\end{cases}
$$
The next step is to add the midpoints of this vector of elements, hence we define the set of unbounded nodes by
$$
x_j^{(k)}(i) = \begin{cases}
\phi_j^{(k)}((i+1)/2) & i=1,3,\dots,6r-1\\
(\phi_j^{(k)}(i-1)+\phi_j^{(k)}(i+1))/2 & i=2,4,\dots,6r-2\\
\end{cases}.
$$
This set of nodes can be used for integration of a function over the interval $(-\infty, \infty)$. For any $k=1,\dots,K$, the integral is unbounded for all $j=1,\dots,p-1.$ However, if the limits of integration are finite in a single dimension, then the set of nodes must be truncated in order to obtain an accurate evaluation of the integral. This is the case for the range $[\bar{y}_s,\bar{y}_{s+1}]$ when at least one of these limits is finite. In such a case, we take the set of nodes $x_p^{(k)}(1),\dots,x_p^{(k)}(6r-1)$, trim off any values outside of 
$[\bar{y}_s,\bar{y}_{s+1}]$  and introduce new points at $\bar{y}_s$ and $\bar{y}_{s+1}$ if necessary. Let $m_j^{(k)}$ denote the number of nodes at analysis $k$ in dimension $j$ and note that $m_j^{(k)}=6r-1$ for $j=1,\dots,p-1$ but $m_p^{(k)}$ is to be determined. Then, applying Simpson's numerical integration rule in one dimension, the weights associated with these nodes are given by
$$
w(i_j^{(k)}) = \begin{dcases}
\tfrac{1}{6}(x_j^{(k)}(3)-x_j^{(k)}(1)) & i_j^{(k)}=1\\
\tfrac{1}{6}(x_j^{(k)}(i_j^{(k)}+2)-x_j^{(k)}(i_j^{(k)}-2) & i_j^{(k)}=3,5,\dots,m_j^{(k)}-2\\
\tfrac{4}{6}(x_j^{(k)}(i_j^{(k)}+1)-x_j^{(k)}(i_j^{(k)}-1) & i_j^{(k)}=2,4,\dots,m_j^{(k)}-1\\
\tfrac{1}{6}(x_j^{(k)}(m_j^{(k)})-x_j^{(k)}(m_j^{(k)}-2) & i_j^{(k)}=m_j^{(k)}
\end{dcases}.
$$

One could consider placing the nodes around the conditional mean given previous dimensions. This would lead to the set of grid points being orthogonal to the covariance matrix. We perform a sensitivity analysis to compute the increase in accuracy of the probability functions when the placement of the points reflects the correlation between the endpoints. For this example we consider the function $\Delta(\bs{\theta})$ in Equation~(1) of the main text and we consider the effects under the null with an extreme correlation for the worst case scenario consideration. Hence, the parameter values we use for the sensititvity analysis are given by
$$
\bs{\theta}_0=(0,0), \Sigma^{(1)}=\left(\begin{array}{cc} 1 & 0.99 \\ 0.99 & 1 \end{array}\right)
$$
Table~\ref{tbl:corr} shows the results of the sensitivity analysis for a range of values of $b^{(1)}$. When adjusting the gridpoints to account for the correlation, we need fewer gridpoints than if we do not adjust (as we have done in the main text). However, note that the un-adjusted version still performs well even for this extreme case where endpoints are correlated with $\rho=0.99$. Adjusting for the correlation adds an extra calculation at every univariate grid point so the computation time roughly doubles. With $p=2$ and $K=5$, this is an extra $(12r-1)^{2p}$ calculations. We find that the extra computational effort required is not worth the gain in accuracy and the un-adjusted version is efficient for moderate correlation levels. 
\begin{table}
  \begin{center}
  \begin{tabular}{l|cc|cc}
  \hline
  & \multicolumn{2}{c}{Un-adjusted} & \multicolumn{2}{c}{Adjusted for correlation} \\
  $b^{(1)}$& Property & $\hat{\psi}^{(1)}$ &
  Property & $\hat{\psi}^{(1)}$\\\hline
  \rule{0pt}{4ex}-1.6292&13&0.89999&11&0.89998\\-1.6292&14&0.90000&12&0.90000\\-1.6292&15&0.90000&13&0.90000\\\rule{0pt}{4ex}-3.8173&14&0.98999&11&0.97499\\-3.8173&15&0.99000&12&0.97500\\-3.8173&16&0.99000&13&0.97500\\\rule{0pt}{4ex}-5.3799&14&0.98999&9&0.99001\\-5.3799&15&0.99000&10&0.99000\\-5.3799&16&0.99000&11&0.99000\\
  \hline
  \end{tabular}
  \caption{Sensitivity analysis comparing probability when adjusting grid points in Simpson's
  multivariate numerical integration rule to be orthogonal to covariance matrix. For this extreme case,
  the correlation between the endpoints is $\rho=0.99$.\label{tbl:corr}}
  \end{center}
  \end{table}
\section{Approximate distribution of sequential global statistics using the Delta Method}
In the main text, the Delta method was used to find the approximate distribution of the global summary statistic. We stated that the global summary statistic $\Delta(\hat{\bs{\theta}}^{(k)})$ is approximately normally distributed. For a fixed sample trial, this is known to hold as proved by~\citet{doob1935limiting}. The extension to proving that the canonical joint distribution holds for the sequence $\Delta(\hat{\bs{\theta}}^{(1)}),\dots,\Delta(\hat{\bs{\theta}}^{(K)})$ follows similar arguments and in turn we find that we only need to prove that the independent increments property holds. The following Theorem shows that the distributional results referred to in Section~3.2 hold.
 \begin{theorem}
Let $\bs{\theta}$ be a $p\times 1$  vector of parameters in a statistical model. Suppose that $\hat{\bs{\theta}}^{(k)}$ is the parameter estimate for $\bs{\theta}$ found at analysis $k$ of a group sequential trial with $K$ analyses and that $\Sigma^{(k)}$ is the variance-covariance matrix for $\hat{\bs{\theta}}^{(k)}$. Let $\Delta(\bs{\theta})$ be a function which returns a scalar output and suppose that the sequence $\hat{\bs{\theta}}^{(1)}, \dots, \hat{\bs{\theta}}^{(K)}$ has the canonical joint distribution (CJD). Then the CJD holds approximately for the sequence of estimates $\Delta(\hat{\bs{\theta}}^{(1)}), \dots, \Delta(\hat{\bs{\theta}}^{(K)})$. That is
\begin{enumerate}
\item $\Delta(\hat{\bs{\theta}}^{(1)}), \dots, \Delta(\hat{\bs{\theta}}^{(K)})$ is multivariate normally distributed
\item $\Delta(\hat{\bs{\theta}}^{(k)}) \sim 
N\left( \Delta(\bs{\theta}), \left[\frac{\partial\Delta(\bs{\theta})}{\partial \bs{\theta}}\right]^T \Sigma^{(k)} \left[\frac{\partial\Delta(\bs{\theta})}{\partial \bs{\theta}}\right]\right)$ for $1\leq k\leq K$
\item $Cov\left(\Delta(\hat{\bs{\theta}}^{(k_1)}),\Delta(\hat{\bs{\theta}}^{(k_2)})\right) = \left[\frac{\partial\Delta(\bs{\theta})}{\partial \bs{\theta}}\right]^T\Sigma^{(k_2)}\left[\frac{\partial\Delta(\bs{\theta})}{\partial \bs{\theta}}\right]$ for $k=1\leq k_1 \leq k_2 \leq K.$
\end{enumerate}
\end{theorem}

\begin{proof}
The proof of this theorem uses the Taylor expansion of the function $\Delta(\hat{\bs{\theta}})$ around $\bs{\theta}$ which is given by
$$
\Delta(\hat{\bs{\theta}})=
\Delta(\bs{\theta})+ \left[\frac{\partial}{\partial\bs{\theta}}\Delta(\bs{\theta})\bigg\rvert_{\bs{\theta}=\bs{\theta}^*}\right]^T(\hat{\bs{\theta}}-\bs{\theta})
$$
where $\bs{\theta}^*$ lies on the line segment between $\bs{\theta}$ and $\hat{\bs{\theta}}$. For each $k=1,\dots, K,$ we shall apply this Taylor expansion at the point $\hat{\bs{\theta}}=\hat{\bs{\theta}}^{(k)}.$ Therefore, for each $k=1,\dots,K$, we have
$$
\Delta(\hat{\bs{\theta}}^{(k)})=\Delta(\bs{\theta})+ \Delta'(\bs{\theta}^{*(k)})^T(\hat{\bs{\theta}}^{(k)}-\bs{\theta}) $$
where $\bs{\theta}^{*(k)}$ lies on the line segment between $\bs{\theta}$ and $\hat{\bs{\theta}}^{(k)}$ and $\Delta'(\bs{\theta}^{*(k)})=\partial/\partial\bs{\theta} \left(\Delta(\bs{\theta})|_{\bs{\theta}=\bs{\theta}^{*(k)}}\right)$ is shorthand notation. For each $k=1,\dots,K$, the parameter estimate $\hat{\bs{\theta}}^{(k)}$ is consistent for $\bs{\theta}$ and since $\bs{\theta}^{*(k)}$ lies on the line segment between $\bs{\theta}$ and $\hat{\bs{\theta}}^{(k)}$, the difference between $\bs{\theta}^{*(k)}$ and $\bs{\theta}$ is asymptotically negligible. Therefore for each $k=1,\dots,K$ we have approximately
\begin{equation}
\label{eq:taylor_delta1}
\Delta(\hat{\bs{\theta}}^{(k)})=\Delta(\bs{\theta})+ \Delta'(\bs{\theta})^T(\hat{\bs{\theta}}^{(k)}-\bs{\theta}).
\end{equation}

The proof of conditions 1 and 2 are standard results given by~\citet{doob1935limiting} who present the multivariate version of the Delta method. The main idea follows by stacking Equations~\eqref{eq:taylor_delta1} for each $k=1,\dots,K.$ Then, Slutsky's Theorem can be applied to the vector of stacked equations. It remains to prove property 3. Using the approximation in Equation~\eqref{eq:taylor_delta1}, the covariance is given by
\begin{align*}
&Cov\left(\Delta(\hat{\bs{\theta}}^{(k_1)}),\Delta(\hat{\bs{\theta}}^{(k_2)})\right) \\
=&Cov\left(\Delta(\bs{\theta})+ \Delta'(\bs{\theta})^T(\hat{\bs{\theta}}^{(k_1)}-\bs{\theta}) , \Delta(\bs{\theta})+ \Delta'(\bs{\theta})^T(\hat{\bs{\theta}}^{(k_2)}-\bs{\theta})\right) \\
=&\Delta'(\hat{\bs{\theta}}^{(k_1)})^TCov\left(\hat{\bs{\theta}}^{(k_1)},\hat{\bs{\theta}}^{(k_2)}\right)\Delta'(\hat{\bs{\theta}}^{(k_2)}).
\end{align*}

By property 3 of the CJD for the sequence $\hat{\bs{\theta}}^{(1)},\dots,\hat{\bs{\theta}}^{(K)}$, we have that $Cov(\hat{\bs{\theta}}^{(k_1)},\hat{\bs{\theta}}^{(k_2)})=\Sigma^{(k_2)}$ and we see the result
$$
Cov\left(\Delta(\hat{\bs{\theta}}^{(k_1)}),\Delta(\hat{\bs{\theta}}^{(k_2)})\right) =\Delta'(\hat{\bs{\theta}}^{(k_1)})^T\Sigma^{(k_2)}\Delta'(\hat{\bs{\theta}}^{(k_2)}).
$$
\end{proof}

In the main text, we give theoretical arguments that the Delta method gives a poor approximation for a case where the global summary statistic is not linear. We now provide evidence, via simulation, to support this claim. The following example has parameter values
\begin{equation}
\label{eq:param_vals}
\bs{\theta}_0=(0,0), \bs{\theta}_A = (1.625,1.625), \Sigma^{(k)} =\frac{1}{n^{(k)}} \left(\begin{array}{cc} 40 & 10 \\ 10 & 40 \end{array}\right).
\end{equation}
For this example with $\Delta(\bs{\theta})$ given in Equation~(1) of the main text, we note that using the Delta method can result in an estimator with variance equal to zero since $\partial\Delta(\bs{\theta})/\partial\bs{\theta}=\mathbf{0}$ under $H_0$. We can however use the value of $\partial\Delta(\bs{\theta})/\partial\bs{\theta}$ under $H_A$ in place of the equivalent under $H_0.$ Hence, the Delta method applies the approximate marginal distributions $\Delta(\hat{\bs{\theta}}^{(k)})\sim N(0, 264/n^{(k)})$ under $H_0$ and $\Delta(\hat{\bs{\theta}}^{(k)})\sim N(2.64, 264/n^{(k)})$ under $H_A.$ Figure~\ref{fig:norm_delta} shows the Delta Method approximation for this example at the first analysis (or in a fixed sample trial) with $n^{(1)}=122$ patients per treatment arm. For each histogram, $10^5$ observations of $\Delta(\hat{\bs{\theta}}^{(1)})$ are generated. The estimates do not appear to be normally distributed and this is more extreme under $H_0.$ The red lines show the probability density functions of the approximated distributions calculated using the Delta method and confirm our suggestion that the estimates are not normally distributed. The problem is particularly relevant in the tails and due to the nature of trial design, this may lead to poor control of error rates.
\begin{figure}[t]
\centering\includegraphics[width=\textwidth]{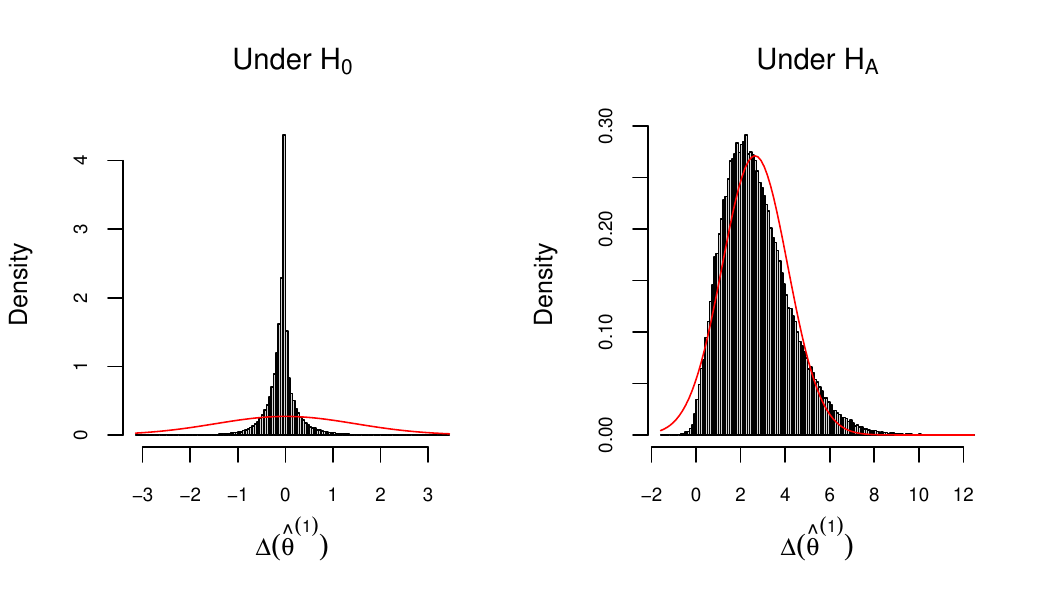}
\caption{Histograms of $\Delta(\hat{\bs{\theta}}^{(1)})$ with $\Delta(\bs{\theta})$ given by Equation~(1) of the main text. Simulated under $H_0$ with $n^{(1)}=122$ and parameter values given by~\eqref{eq:param_vals}. There are $10^5$ replicates per histogram. Red lines show the approximated normal distribution when using the delta method.}
\label{fig:norm_delta}
\end{figure}

\section{Constructing boundary constants using non-parametric statistics.}
We have briefly described a method for estimating probabilities and finding boundary constants using Monte Carlo simulation and we now describe this method in detail. Let $N$ be the number of Monte Carlo samples under the null and under the alternative and suppose that replicates of the sequence of parameter estimates are generated independently from the following distributions
\begin{align*}
\{\hat{\bs{\theta}}_0^{(1)},\dots,\hat{\bs{\theta}}_0^{(K)}\}_i
\sim
MCJD_{\bs{\theta}_0} \\
\{\hat{\bs{\theta}}_A^{(1)},\dots,\hat{\bs{\theta}}_A^{(K)}\}_i
\sim
MCJD_{\bs{\theta}_A}
\end{align*}
for $i=1,\dots, N$ where $MCJD_{\bs{\theta}}$ refers to the multivariate canonical joint distribution with mean true vector $\bs{\theta}$. It should be noted that the subscript notation denoted $i$ here refers to Monte Carlo replicates each of dimension $p$ and should not be confused with the notation $\hat{\theta}^{(k)}_j$ which represents the $j^{th}$ dimension of the vector. Then the stage-wise type 1 and type 2 error rates can be estimated as
\begin{align*}
\psi^{(k)}
\approx
\frac{1}{N}\sum_{i=1}^{N}
\mathbb{1}_{\mathcal{C}^{(1)}}(\{\hat{\bs{\theta}}_0^{(1)}\}_i)\dots
\mathbb{1}_{\mathcal{C}^{(k-1)}}(\{\hat{\bs{\theta}}_0^{(k-1)}\}_i) 
\mathbb{1}_{\mathcal{R}^{(k)}}(\{\hat{\bs{\theta}}_0^{(k)}\}_i)\\
\xi^{(k)}
\approx
\frac{1}{N}\sum_{i=1}^{N}
\mathbb{1}_{\mathcal{C}^{(1)}}(\{\hat{\bs{\theta}}_A^{(1)}\}_i)\dots
\mathbb{1}_{\mathcal{C}^{(k-1)}}(\{\hat{\bs{\theta}}_A^{(k-1)}\}_i)
\mathbb{1}_{\mathcal{A}^{(k)}}(\{\hat{\bs{\theta}}_A^{(k)}\}_i).
\end{align*}

Supposing that 
$$
(\bar{\Delta}^{(k)})_1 < (\bar{\Delta}^{(k)})_2<\dots<(\bar{\Delta}^{(k)})_N
$$
are the ordered samples of $\Delta((\hat{\bs{\theta}}^{(k)})_1),\dots,\Delta((\hat{\bs{\theta}}^{(k)})_N)$ then by the Inverse Transform Theorem, we can show that
$$
\mathbb{P}_{\bs{\theta}}(\Delta(\hat{\bs{\theta}}^{(k)}) \leq \frac{i}{N})
\approx
(\bar{\Delta}^{(k)})_i
\hspace{1cm} \text{and} \hspace{1cm}
\mathbb{P}_{\bs{\theta}}(\Delta(\hat{\bs{\theta}}^{(k)}) > \frac{N-i}{N})
\approx
(\bar{\Delta}^{(k)})_i.
$$
Therefore, we can find boundary constants according to the following scheme:

\hspace{0.5cm} For $k=1$

\hspace{1cm} Sample $(\bs{\hat{\theta}}^{(1)})_i$ under $H_0$ for $i=1,\dots,N$

\hspace{1cm} Set $b^{(1)}$ as the value of $i/N$ such that $(\bar{\Delta}^{(1)})_i$ is closest to $\psi^{(1)}$

\hspace{1cm} Sample $(\bs{\hat{\theta}}^{(1)})_i$ under $H_A$ for $i=1,\dots,N$

\hspace{1cm} Set $a^{(1)}$ as the value of $(N-i)/N$ such that $(\bar{\Delta}^{(k)})_i$ is closest to $\xi^{(1)}$

\hspace{0.5cm} For analyses $k=2,\dots,K$

\hspace{1cm} Keep all $(\bs{\hat{\theta}}^{(k-1)})_i$ in the interval $[a^{(k-1)},b^{(k-1)}]$ which were sampled under $H_0$ and call the length $N_0^{(k)}$

\hspace{1cm} Given $(\bs{\hat{\theta}}^{(k-1)})_i$, sample $(\bs{\hat{\theta}}^{(k)})_i$ under $H_0$ for $i=1,\dots,N_0^{(k)}$

\hspace{1cm} Set $b^{(k)}$ as the value of $i/N_0^{(k)}$ such that $(\bar{\Delta}^{(k)})_i$ is closest to $\psi^{(k)}$

\hspace{1cm} Keep all $(\bs{\hat{\theta}}^{(k-1)})_i$ in the interval $[a^{(k-1)},b^{(k-1)}]$ which were sampled under $H_A$ and call the length $N_A^{(k)}$

\hspace{1cm} Given $(\bs{\hat{\theta}}^{(k-1)})_i$, sample $(\bs{\hat{\theta}}^{(k)})_i$ under $H_A$ for $i=1,\dots,N_A^{(k)}$

\hspace{1cm} Set $a^{(k)}$ as the value of $(N_A^{(k)}-i)/N_A^{(k)}$ such that $(\bar{\Delta}^{(k)})_i$ is closest to $\xi^{(k)}.$
\section{Sensitivity analyses for misspecification of nuisance parameter}
During the sample size calculation, it is necessary to make assumptions about all entries of the matrix $M$, which is the multivariate nuisance parameter. We believe there to be a reasonable amount of knowledge about the structure and values of $M$, but we shall perform sensitivity analyses to observe the worst case scenario. In what follows, the \emph{true} nuisance matrix will be given by
$$ 
M=\left(\begin{array}{cc} 40 & 40\rho \\ 40\rho & 40 \end{array}\right)
$$
where $\rho$ is the correlation parameter and can take values between -1 and 1. For the sensitivity analysis, we shall calculate the sample size required for a MGST with linear global summary statistic given by $\Delta(\bs{\theta})=\theta_1+\theta_2$, significance level $\alpha=0.025$ at $\bs{\theta}_0=(0,0)$, power $1-\beta=0.9$ at $\bs{\theta}_A=(1.625,1.625)$ and error spending boundaries given by $\pi_1(t)=\max\{\alpha t^2, \alpha\}$ and $\pi_2(t)=\max\{\beta t^2, \beta\}$ and we shall do so for various misspecified values of the correlation parameter, denoted $\tilde{\rho}.$ Following the methods described for calculating the sample size, when $\tilde{\rho}$ takes values of $-0.5,-0.25,0,0.025,0.5$ the required sample size, $n^{(k)}$ will be $9k,14k,18k,22k$ and $27k$ respectively. Next, for each of the values of $\tilde{\rho},$ we calculated the obtained power for all possible true correlations.

Figure~\ref{fig:sens_ss} shows the results of this sensitivity analysis. First we see that in general, power increases with misspecified correlation $\tilde{\rho}$ due to the increase in sample size. In line with this finding, power decreases as true correlation $\rho$ increases. The points in Figure~\ref{fig:sens_ss} are for reference, showing power 0.9 at the true value of $\rho.$ In each case, observed power is slightly higher than 0.9 which confirms the accuracy of the sample size calculation accounting for the discrete nature. The cyclic, non-smooth, nature of the graph occurs because changes in information levels result in a different number of analyses. By design, the trial stops at the smallest $k$ such that $\mathcal{I}^{(k)} \geq \mathcal{I}_{max}$. As the magnitude of $\rho$ increases, the values of $\mathcal{I}^{(k)}=40\sqrt{1-\rho^2}$ decrease for $k=1,\dots,K$ and hence, the final analysis changes with the value of $\rho.$ For this particular choice of $\Delta(\bs{\theta})$ and $\bs{\theta}_A,$ we are only worse-off if we have underestimated the correlation, in which case power will be lower than planned. An example is that we plan the trial with $\tilde{\rho}=-0.5$ and observe $\rho=0.5$ which leads to observed power of roughly 0.5. However, we highlight the importance of thorough planning in this case since a belief of $\rho=-0.5$ does not align with the choice of null region $\mathcal{N}=\{\bs{\theta};\theta_1 \leq 0 ,\theta_2 \leq 0\}$ and point alternative $\bs{\theta}_A=(1.625,1.625)$. The recommendation is therefore to ensure that assumptions truly reflect beliefs about the endpoints and also to calculate the required sample size for a range of nuisance matrices, overestimating the magnitude of correlation if we wish to be conservative.

\begin{figure}[t]
\centering\includegraphics[width=\textwidth]{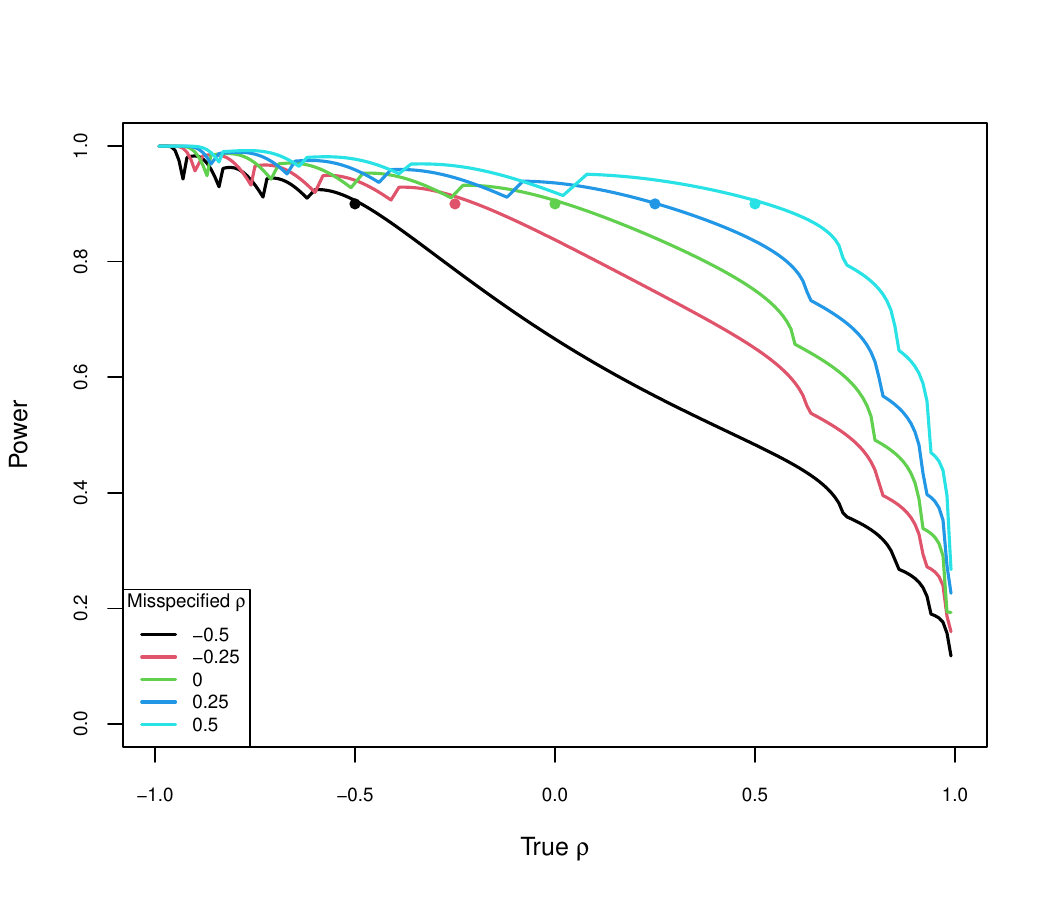}
\caption{Sensitivity analysis for misspecifying the matrix $M$ at the design stage of the trial showing the effects on obtained power. The points show the obtained power compared with planned $1-\beta=0.9$ when the misspecified correlation equals the true correlation.}
\label{fig:sens_ss}
\end{figure}
\section{Further simulation results}
Tables~\ref{tbl:fixed} and~\ref{tbl:gst} provide the results of additional simulation studies. These results are for the readers reference and the main features are discussed in the main text.

\begin{table}
  \begin{center}
  \begin{tabular}{l|lccc|lccc}
  \hline
  &\multicolumn{4}{c}{$\Delta(\bs{\theta})=\theta_1+\theta_2$}&
  \multicolumn{4}{c}{$\Delta(\bs{\theta})=\theta_1\theta_2$}\\
  \cmidrule{2-5}\cmidrule{6-9}
  Method &
  Property &$b^{(1)}$&$\psi^{(1)}$ & $\xi^{(1)}$&
  Property &$b^{(1)}$&$\psi^{(1)}$ & $\xi^{(1)}$\\\hline
  \rule{0pt}{4ex}Simpson & $r=9$&2.3032&0.05000&0.24947&
    $r=9$&1.0680&0.05001&0.24994\\Simpson & $r=10$&2.3032&0.05000&0.24948&
    $r=10$&1.0681&0.05000&0.24995\\Simpson & $r=11$&2.3032&0.05000&0.24948&
    $r=11$&1.0681&0.05000&0.24995\\\rule{0pt}{4ex}Delta & $r=7$&2.3032&0.05001&0.24947&
    $r=6$&3.6713&0.00205&0.69863\\Delta & $r=8$&2.3032&0.05000&0.24948&
    $r=7$&3.6714&0.00205&0.69865\\Delta & $r=9$&2.3032&0.05000&0.24948&
    $r=8$&3.6714&0.00205&0.69865\\\rule{0pt}{4ex}Monte Carlo & $N=10^5$&2.3117&0.04938&0.25141&
    $N=10^5$&1.0633&0.05032&0.24902\\Monte Carlo & $N=10^6$&2.3071&0.04972&0.25036&
    $N=10^6$&1.0662&0.05013&0.24959\\Monte Carlo & $N=10^7$&2.3021&0.05009&0.24921&
    $N=10^7$&1.0687&0.04996&0.25008\\\rule{0pt}{4ex}Simpson & $r=16$&1.8533&0.00100&0.00993&
    $r=14$&0.8027&0.00100&0.00981\\Simpson & $r=17$&1.8533&0.00100&0.00994&
    $r=15$&0.8028&0.00100&0.00981\\Simpson & $r=18$&1.8533&0.00100&0.00994&
    $r=16$&0.8028&0.00100&0.00981\\\rule{0pt}{4ex}Delta & $r=11$&1.8533&0.00100&0.00993&
    $r=14$&2.9850&0.00000&0.65372\\Delta & $r=12$&1.8533&0.00100&0.00994&
    $r=15$&2.9850&0.00000&0.65373\\Delta & $r=13$&1.8533&0.00100&0.00994&
    $r=16$&2.9850&0.00000&0.65373\\\rule{0pt}{4ex}Monte Carlo & $N=10^5$&1.8530&0.00100&0.00992&
    $N=10^5$&0.7833&0.00113&0.00894\\Monte Carlo & $N=10^6$&1.8436&0.00106&0.00951&
    $N=10^6$&0.8025&0.00100&0.00980\\Monte Carlo & $N=10^7$&1.8555&0.00099&0.01003&
    $N=10^7$&0.8066&0.00098&0.00999\\
  \hline
  \end{tabular}
  \caption{Comparison of methods for constructing a fixed sample trial using the linear function
  $\Delta(\theta_1,\theta_2)=\theta_1+\theta_2$ and non-linear function 
  $\Delta(\theta_1,\theta_2)$ given in Equation~(1) of the main text.
  True observed probabilities $\psi^{(1)}$ and $\xi^{(1)}$ were calculated using
  the Delta method with r=128 for the linear case and Simpson's method with
  $r=32$ for the non-linear case.
  \label{tbl:fixed}}
  \end{center}
  \end{table}
\begin{table}
  \begin{center}
  \begin{tabular}{lll|cccc|cccc}
  \hline
  &&&\multicolumn{4}{c}{$\Delta(\bs{\theta})=\theta_1+\theta_2$}&
  \multicolumn{4}{c}{$\Delta(\bs{\theta})=\theta_1\theta_2$}\\
  \cmidrule{4-7}\cmidrule{8-11}
  Method & Property & $k$ &
  $a^{(k)}$ & $b^{(k)}$ & $\psi^{(k)}$ & $\xi^{(k)}$&
  $a^{(k)}$ & $b^{(k)}$ & $\psi^{(k)}$ & $\xi^{(k)}$\\\hline
  \rule{0pt}{4ex}Simpson&$r=6$&1&-0.8576&5.1503&0.00501&0.02000&-1.7454&5.9367&0.00500&0.02000\\Simpson&$r=6$&2&0.5430&3.5239&0.00500&0.02000&-0.3206&2.7653&0.00500&0.02000\\Simpson&$r=6$&3&1.1873&2.7791&0.00500&0.02000&0.0953&1.7089&0.00500&0.02000\\Simpson&$r=6$&4&1.5894&2.3155&0.00500&0.02000&0.4442&1.1846&0.00500&0.02000\\Simpson&$r=6$&5&1.9241&1.9241&0.00500&0.01895&0.8337&0.8337&0.00500&0.02247\\\rule{0pt}{4ex}Delta&$r=128$&1&-0.8575&5.1517&0.00500&0.02000&-3.9044&8.2089&0.00134&0.00235\\Delta&$r=128$&2&0.5431&3.5239&0.00500&0.02000&-1.6729&5.6156&0.00017&0.00132\\Delta&$r=128$&3&1.1873&2.7790&0.00500&0.02000&-0.6513&4.4355&0.00004&0.00197\\Delta&$r=128$&4&1.5895&2.3154&0.00500&0.02000&-0.0351&3.7258&0.00002&0.00566\\Delta&$r=128$&5&1.9239&1.9239&0.00500&0.01893&3.2391&3.2391&0.00001&0.58468\\\rule{0pt}{4ex}Monte Carlo&$N=10^7$&1&-0.8741&5.1409&0.00508&0.01960&-1.7450&5.9327&0.00501&0.02001\\Monte Carlo&$N=10^7$&2&0.5522&3.5277&0.00495&0.02045&-0.4068&2.8606&0.00441&0.01609\\Monte Carlo&$N=10^7$&3&1.1653&2.7705&0.00514&0.01881&-0.0448&1.8957&0.00343&0.01076\\Monte Carlo&$N=10^7$&4&1.5824&2.3045&0.00516&0.01999&0.1098&1.5332&0.00188&0.00601\\Monte Carlo&$N=10^7$&5&1.9418&1.9418&0.00465&0.02019&1.4269&1.4269&0.00069&0.10868\\\rule{0pt}{4ex}Simpson&$r=6$&1&-0.8576&5.1503&0.00501&0.02000&-1.7454&5.9367&0.00500&0.02000\\Simpson&$r=6$&2&0.5430&3.5239&0.00500&0.02000&-0.3206&2.7653&0.00500&0.02000\\Simpson&$r=6$&3&1.1873&2.7791&0.00500&0.02000&0.0953&1.7089&0.00500&0.02000\\Simpson&$r=6$&4&1.5894&2.3155&0.00500&0.02000&0.4442&1.1846&0.00500&0.02000\\Simpson&$r=6$&5&1.9241&1.9241&0.00500&0.01895&0.8337&0.8337&0.00500&0.02247\\\rule{0pt}{4ex}Delta&$r=128$&1&-0.8575&5.1517&0.00500&0.02000&-3.9044&8.2089&0.00134&0.00235\\Delta&$r=128$&2&0.5431&3.5239&0.00500&0.02000&-1.6729&5.6156&0.00017&0.00132\\Delta&$r=128$&3&1.1873&2.7790&0.00500&0.02000&-0.6513&4.4355&0.00004&0.00197\\Delta&$r=128$&4&1.5895&2.3154&0.00500&0.02000&-0.0351&3.7258&0.00002&0.00566\\Delta&$r=128$&5&1.9239&1.9239&0.00500&0.01893&3.2391&3.2391&0.00001&0.58468\\\rule{0pt}{4ex}Monte Carlo&$N=10^7$&1&-0.8741&5.1409&0.00508&0.01960&-1.7450&5.9327&0.00501&0.02001\\Monte Carlo&$N=10^7$&2&0.5522&3.5277&0.00495&0.02045&-0.4068&2.8606&0.00441&0.01609\\Monte Carlo&$N=10^7$&3&1.1653&2.7705&0.00514&0.01881&-0.0448&1.8957&0.00343&0.01076\\Monte Carlo&$N=10^7$&4&1.5824&2.3045&0.00516&0.01999&0.1098&1.5332&0.00188&0.00601\\Monte Carlo&$N=10^7$&5&1.9418&1.9418&0.00465&0.02019&1.4269&1.4269&0.00069&0.10868\\
  \hline
  \end{tabular}
  \caption{Comparison of methods for constructing group sequential trials using the linear function
  $\Delta(\theta_1,\theta_2)=\theta_1+\theta_2$ and the non-linear function
  $\Delta(\theta_1,\theta_2)$ given by Equation~(1) if the main text.
  True observed probabilities $\psi^{(k)}$ and $\xi^{(k)}$ were calculated using
  the Delta method with $r=128$ for the linear case and Simpson's Method with $r=16$
  for the non-linear case.
  \label{tbl:gst}}
  \end{center}
  \end{table}

\bibliographystyle{plainnat}  
\bibliography{enrichment}